\documentclass[preprint,numberswithinsect]{elsarticle}
\newif\iflncs

\newcommand{\lncsqed}{\iflncs\hfill $\Box$\fi}

\usepackage[utf8]{inputenc}
\usepackage[T1]{fontenc}
\usepackage{amssymb,amsmath,amsthm}
\usepackage{xspace}
\usepackage{tikz}
\usepackage{centernot}
\usepackage[ruled,noend,linesnumbered]{algorithm2e}
\usepackage{hyperref}
\usepackage{varioref}
\usepackage{subcaption}
\usepackage{cleveref}
\usepackage[margin=1.6in]{geometry}

\newcommand{\prob}[3]{\begin{quote}  \textsc{#1}\\  \textbf{Input:} #2\\  \textbf{Question:} #3\end{quote}}

\usepackage[]{todonotes}

\newcommand{\from}{\ensuremath{\leftarrow}}

\newtheorem{nclaim}{Claim}
\newtheorem{erule}{Rule}

\newcommand{\G}{\ensuremath{\mathcal{G}}}

\newcommand{\TC}{\textsf{TC}\xspace}
\newcommand{\sTC}{\textsf{strict-TC}\xspace}

\newcommand{\TST}{\textsc{TC Spanning Tree}\xspace}

\newcommand{\BS}{\textsc{Bidirectional Spanner}\xspace}
\newcommand{\BC}{\textsc{Bidirectional Connectivity}\xspace}

\newcommand{\Oh}{\mathcal{O}}
\newcommand{\fes}{\ensuremath{\mathrm{fes}}\xspace}
\newcommand{\mc}{\mathcal{C}}

\usetikzlibrary{arrows,shapes, graphs, graphs.standard, positioning,calc}
\usetikzlibrary{decorations.pathreplacing, decorations.pathmorphing,fit}
\usetikzlibrary{patterns,patterns.meta}

\newtheorem{theorem}{Theorem}
\newtheorem{lemma}[theorem]{Lemma}
\newtheorem{corollary}[theorem]{Corollary}

\newtheorem{definition}[theorem]{Definition}
\newtheorem{remark}{Remark}
\newtheorem{question}{Question}

\begin{document}

\begin{frontmatter}

\title{In Search of the Lost Tree: Hardness and relaxation of spanning trees in temporal graphs\tnoteref{t1}\tnoteref{t2}}

\tnotetext[t1]{A preliminary version of this work was presented at the 31st Int. Colloquium on Structural Information and Communication Complexity (SIROCCO 2024).}
\tnotetext[t2]{Supported by French ANR, project TEMPOGRAL (ANR-22-CE48-0001) and Swiss NSF, project RECAPT (200021-236640).}

\author[1]{Arnaud Casteigts}

\author[2]{Timothée Corsini}

\author[2,3]{Nils Morawietz}

\affiliation[1]{organization={CS Department, University of Geneva}, city={Geneva}, country={Switzerland}}
\affiliation[2]{organization={LaBRI, University of Bordeaux}, city={Bordeaux}, country={France}}
\affiliation[3]{organization={Institute of Computer Science, Friedrich Schiller University Jena}, city={Jena}, country={Germany}}


\begin{abstract}
  A temporal graph is a graph whose edges appear at certain
  points in time. These graphs are temporally connected (in class
  \TC) if all vertices can reach each other by
  temporal paths (traversing the edges in chronological order).
  Reachability based on temporal paths is not transitive, with
  important consequences. For instance, \TC graphs do not
  always admit \TC spanning trees.

  In this paper, we show that deciding if a given temporal graph
  admits a \TC spanning tree is actually NP-complete. Then, we explore
  possible relaxations. A key feature of \TC spanning trees is to
  support reachability along the same paths in both directions. We
  show that this property is not equivalent to \TC spanning trees, it
  is more general and it can be tested in polynomial time. Still,
  minimizing the size of a spanner preserving this property---a bidirectional spanner---is
  \textsf{NP}-hard even more generally than \TC spanning tree, including the setting of simple temporal graphs.

  Along the way, we show that deciding the existence of \TC spanning tree is FPT
  when parameterized by the feedback edge set number (\fes) of
  the underlying graph, and deciding bidirectional spanners of size $k$ is
  FPT when parameterized by \fes + $\ell$ (the maximum number of
  labels per edge). On the structural side, we show that 
  \TC trees always admit a pivot vertex or a pivot
  edge---reachable by all vertices by a certain time and able to reach all vertices afterward---a fact that may be of independent interest.
\end{abstract}


\end{frontmatter}

\section{Introduction}
Temporal graphs are appropriate models for capturing
time-varying phenomena in transportation, social networks, biology, robotics,
scheduling, and distributed computing. In the basic model, a temporal graph
can be represented by a labeled graph
$\G=(G,\lambda)$, where $G=(V,E)$ is a standard finite graph called the footprint (undirected in this work), and
$\lambda : E \to 2^{\mathbb{N}}\setminus \emptyset$ is a function that assigns one or several
time labels to each edge of $E$, interpreted as discrete presence times.
A central concept in these graphs is the one of temporal path (or journey), which is a
sequence $\langle(e_i, t_i)\rangle$ such that
$\langle e_i \rangle$ is a path in $G$, $t_i \in \lambda(e_i)$, and $\langle t_i \rangle$ is
non-decreasing. Such a path is called \emph{strict} if $\langle t_i \rangle$ is increasing. A
graph $\G$ is (strictly) temporally connected if
there exists at least one (strict) temporal path between every ordered pair of
vertices. One can also write $\G \in \TC$ or $\G \in \sTC$, seeing these properties as classes of temporal graphs.

Reachability based on temporal paths is not symmetric: the fact that a
node $u$ can reach a node $v$ does not imply that $v$ can reach $u$,
even when the footprint is undirected. It is also not
transitive: the fact that $u$ can reach $v$ and $v$ can reach $w$ does
not imply that $u$ can reach $w$, both facts having deep structural
and algorithmic consequences. Over the past two decades, a
growing body of work was devoted to better understanding temporal
reachability, considered from various perspectives, e.g.
$k$-connectivity and separators~\cite{KKK02,FMN+20,CCMP20}, connected
components~\cite{BF03,AF16,AGMS17,RML20}, feasibility of distributed
tasks~\cite{CFQS12,GCLL15,ADD+21,BT15}, schedule
design~\cite{BCV21,BCV23}, data
structures~\cite{BFJ03,WDCG12,RML20,BACT22}, mitigation~\cite{EMMZ19},
shortest paths~\cite{BFJ03,CHMZ21}, enumeration~\cite{EMM22},
stochastic models~\cite{BCF11,CRRZ21}, flows~\cite{ACG+19,VDPS21}, and
exploration~\cite{IKW14,DDFS20,ES22,FMS13}.

One of the first questions from the seminal work of Kempe,
Kleinberg, and Kumar~\cite{KKK02}, concerns the existence of sparse
spanning subgraphs, also called \emph{temporal spanners}, defined as
subgraphs of the input temporal graph that preserve temporal
connectivity using as few edges as possible. The authors
of~\cite{KKK02} observed that \TC spanning trees do not always exist in
\TC graphs, and more generally, there exist \TC
graphs with $\Theta(n \log n)$ edges, all of which are necessary for
connectivity. Fifteen years later, Axiotis and Fotakis~\cite{AF16}
established a much stronger negative result, showing that there even
exist \TC graphs of size $\Theta(n^2)$ that cannot be
sparsified. In a sense, this result dashed the hope of defining
analogs of spanning trees in temporal graphs. Subsequent research
focused on positive results for special cases. For example, if the
input graph is a complete graph, then spanners with $O(n \log n)$
edges always exists~\cite{CPS19}. If an Erdös-Rényi graph of
parameters $n$ and $p$ is augmented with random time labels, then a
nearly optimal spanner with $2n + o(n)$ edges exists asymptotically
almost surely, as soon as the graph becomes temporally
connected~\cite{CRRZ21}.

On the algorithmic side, Axiotis and Fotakis~\cite{AF16} and Akrida,
Gąsieniec, Mertzios, and Spirakis~\cite{AGMS17} independently showed
that minimizing the size temporal spanners is APX-hard. Special types
of spanners were also investigated, such as spanners with low
stretch~\cite{BDG+22a} and fault-tolerant spanners~\cite{BDG+22b}.

\subsection{Contributions}

In this article, we revisit one of the motivations
of~\cite{KKK02}, questioning the (in)existence of \TC spanning
trees, defined as a spanning tree of the footprint
whose labels (inherited from the original graph)
preserve temporal connectivity. As these structures are not universal,
a natural question is the difficulty of deciding whether a given \TC graph admits
a \TC spanning tree.

It is somewhat surprising that the answer to this question is still
unknown more than two decades after the work of Kempe \textit{et
  al.}~\cite{KKK02}. Our first result is to fill this gap by showing
that deciding this natural property is NP-complete. It is actually
hard in both the strict and the non-strict setting. In general, these
two settings are incomparable, which often creates some confusion in
the temporal graph literature~\cite{CCS24}. Here, we show at once that
the problem is \textsf{NP}-hard in both settings, using a reduction
that rely on a \emph{proper} temporal graphs, where any two adjacent
edges have different labels, thus the distinction between strict and
non-strict temporal paths disappears. (We suggest to regard this type of
reductions as a good practice, whenever it can be achieved, in order to obtain
unified results.)

Next, we investigate how the concept of a \TC spanning tree could be relaxed.
\TC spanning trees are not just minimum-size spanners; they also guarantee that
any two vertices can reach each other using the same underlying path in both directions,
a convenient feature in network routing. Relaxing trees in
this direction, we investigate the concept of bidirectional connectivity, where every two nodes can reach each other through at least one temporal path that uses the same underlying path in both direction, and the associated concept of bidirectional spanner (or \emph{bi-spanner}), that consists of a temporal spanner preserving this property.

We obtain both positive and negative results. On the positive side, we
show that bidirectional connectivity can be decided in polynomial
time. The algorithm relies on a more basic primitive (a
\emph{bi-path}), that finds a bidirectional path between a given pair
of vertices. Then, we examine the complexity of deciding whether a
bidirectional spanner on $k$ edges exists. The previous hardness
result on \TC spanning trees implies that this problem is also
\textsf{NP}-hard in proper temporal graphs, as it corresponds to the
special case that $k=n-1$. However, we show that bidirectional
spanners is also \textsf{NP}-hard in \emph{simple} temporal graphs
(graphs whose edges have a single label), where \TC spanning
trees are easy to decide. Thus, bidirectional spanners are in that sense harder
than \TC spanning trees.

In the second part of the paper, we present fixed-parameter algorithms for both problems. More precisely, we show that both problems can be decided in $2^{\Oh(r)} \cdot n^{\Oh(1)}$ and $\ell^{\Oh(r)} \cdot n^{\Oh(1)}$~time, respectively, where~$\ell$ is the maximum number of labels per edge and~$r$ is the feedback edge set number of the footprint.

Along these results, we make a number of structural observations
related to \TC trees, bidirectional
connectivity, and bidirectional spanners, which may be of independent interest. 
In particular, we show that every \TC tree contains a pivot, that is, a vertex or an edge that can be reached by all vertices by a specific time and reach all vertices back after that time. 
The constructive bi-path primitive we introduce could also find application in routing
and security, as it allows two distant nodes to rely on a same set of intermediate nodes for communication.

\subsection{Organization of the paper}

The paper is organized as follows. In Section~\ref{sec:definitions}, we define the main concepts and questions investigated in the paper. In particular, we define the three considered problems, which are \TST, \BC, and \BS.
In Section~\ref{sec:tst}, we characterize the landscape of tractability for \TST, showing that this problem is NP-complete in proper temporal graphs (yet, tractable in simple temporal graphs).
In Section~\ref{sec:bs}, we present an algorithm that solves \BC in polynomial time.
The algorithm relies on a polynomial time primitive that decides whether there is a bi-path between a given pair of vertices. 
In Section~\ref{sec:kbs}, we show that \BS is NP-complete even in simple temporal graphs. 
In \Cref{sec:fes}, we present our parameterized algorithms for both problems.
Finally, Section~\ref{sec:conclusion} concludes the paper with further remarks and open questions.

\section{Definitions}
\label{sec:definitions}

Given a temporal graph $\G=((V,E),\lambda)$ defined as above, we denote by $N(v)=\{u\in V \mid uv \in E\}$
the neighbors of $v$ in the footprint and
$\Delta=\max_{v\in V}\{|N(v)|\}$ the maximum degree in the footprint.
The terms of nodes and vertices are used interchangeably. The static
graph $G_t=(V,E_t)$ where $E_t=\{e \in E \mid t \in \lambda(e)\}$ is
the \emph{snapshot} of $\G$ at time $t$. A pair $(e, t)$ such that
$e \in E$ and $t \in \lambda(e)$ is a \emph{contact} in $\G$ (also
called a temporal edge). The range of $\lambda$ is $[\tau]$ for some $\tau$ called the \emph{lifetime} of $\G$. 

A temporal graph $\G'=(G',\lambda')$ is a \emph{temporal subgraph}
of $\G$, noted $\G' \subseteq \G$, if $G'=(V',E')$ is a subgraph of
$G=(V,E)$, $\lambda' \subseteq \lambda$, and $\lambda'$ is restricted to
$E'$. If~$V'=V$ and~$\G'$ is temporally connected, then $\G'$ is a \emph{temporal
  spanner} of $\G$. As explained above, temporal reachability is not
transitive in general. However, it remains possible to compose two
temporal paths when the first finishes before the second starts. 
Next, we define the concept of pivot structures that have a special role for temporal spanning trees and directly implies temporal connectivity.

\begin{definition}[Pivot contact]
Let $\G=((V,E),\lambda)$ be a temporal graph. A contact $(e,t)$ of $\G$ is a \emph{strict (non-strict) pivot contact} if for each vertex~$v\in V$
\begin{itemize}
\item there is a strict (non-strict) temporal path that starts at~$v$ and traverses edge~$e$ at time~$t$ and
\item there is a strict (non-strict) temporal path that ends in~$v$ and traverses edge~$e$ at time~$t$.
\end{itemize}
\end{definition}

\begin{definition}[Pivot vertex]
Let $\G=((V,E),\lambda)$ be a temporal graph. A temporal vertex $(p,t)$ of $\G$ is a strict (non-strict) \emph{pivot vertex} if for each vertex~$v\in V$
\begin{itemize}
\item there is a strict (non-strict) temporal path starting from~$v$ and arriving at vertex~$p$ at a time~$t'$ with~$t'<t$ ($t' \leq t$) and
\item there is a strict (non-strict) temporal path that starts at vertex~$p$ at a time~$t'' \geq t$ and arrives at vertex~$v$.
\end{itemize}
\end{definition}

Unless otherwise mentioned, we refer to pivot contacts just as pivots, precising explicitly when we deal with pivot vertices.

\paragraph{Representation in memory}

There are several ways to represent a temporal graph in memory. One option is a sequence of snapshots, each specified as an adjacency matrix or adjacency list. A second option is as a list of triplets of the form $(u,v,t)$. A third option is as a time-augmented adjacency list, that is, an adjacency list whose nested entries contain a list of times. In the analysis of our bi-path algorithm, we assume that, when iterating over the neighbors of a given vertex, finding whether a label exists within a certain time range can be done in time $O(\tau)$, which is the case, e.g., with adjacency lists (this could even be done in time $O(\log \tau)$ in this case) or with sequences of adjacency matrices. For most of our other results, the representation does not matter, as one can convert them to one another in polynomial time.

\subsection{\TC spanning trees}

A \emph{\TC spanning tree} of $\G$ is a temporal spanner of $\G$ 
whose footprint is a tree. Clearly, such trees are not
universal in \TC graphs, as illustrated by the
graph in Figure~\ref{fig:square},
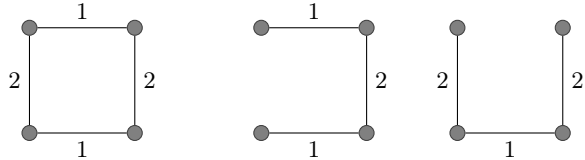
\begin{figure}[h]
  \centering
  \begin{tikzpicture}[scale=1.4]
    \tikzstyle{every node}=[draw=darkgray,fill=gray,circle,inner sep=2pt]
    \path (0,0) node (a){};
    \path (1,0) node (b){};
    \path (1,1) node (c){};
    \path (0,1) node (d){};

    \tikzstyle{every node}=[font=\footnotesize]
    \draw (a)--node[below] {1}(b)--node[right] {2}(c)--node[above] {1}(d)--node[left] {2}(a);
  \end{tikzpicture}~~~~~~~~~~
  \begin{tikzpicture}[scale=1.4]
    \tikzstyle{every node}=[draw=darkgray,fill=gray,circle,inner sep=2pt]
    \path (0,0) node (a){};
    \path (1,0) node (b){};
    \path (1,1) node (c){};
    \path (0,1) node (d){};

    \tikzstyle{every node}=[font=\footnotesize]
    \draw (a)--node[below] {1}(b)--node[right] {2}(c)--node[above] {1}(d);
  \end{tikzpicture}~~~
  \begin{tikzpicture}[scale=1.4]
    \tikzstyle{every node}=[draw=darkgray,fill=gray,circle,inner sep=2pt]
    \path (0,0) node (a){};
    \path (1,0) node (b){};
    \path (1,1) node (c){};
    \path (0,1) node (d){};

    \tikzstyle{every node}=[font=\footnotesize]
    \draw (a)--node[below] {1}(b)--node[right] {2}(c);
    \draw (d)--node[left] {2}(a);
  \end{tikzpicture}
  \caption{\label{fig:square}A temporal graph that is temporally connected (left), in which all temporal subgraphs whose footprint is a tree are \emph{not} temporally connected (middle and right, up to isomorphism).}
\end{figure}
whose underlying spanning trees all induce graphs that are not temporally connected.
Here, a temporal spanning tree would exist if, for example, the two edges labeled $1$ had an additional label of value $3$. We consider the following fundamental problem:
\prob{\TST(TST)}{A temporal graph $\G$.}{Does $\G$ admit a \TC spanning tree?}

\subsection{Bidirectional connectivity}
A bidirectional temporal path (bi-path, for short) between two
vertices $u$ and $v$ is a couple $(p_1,p_2)$ such that $p_1$ is a
temporal path from $u$ to $v$ and $p_2$ is a temporal path from $v$ to
$u$ and both $p_1$ and $p_2$ use the same underlying path (reversed). A temporal graph is bidirectionally connected if all pairs of vertices can reach each other through a bi-path. A bidirectional temporal spanner (bi-spanner) is a temporal spanner that preserves bidirectional connectivity. We
consider the following two additional problems:

\prob{\BC}{A temporal graph $\G$.}{Is $\G$ bidirectionally connected?}

\prob{\BS}{A temporal graph $\G$, an integer $k$.}{Does $\G$ admit a bi-spanner with at most $k$ edges?}

\TC spanning trees are particular cases of bi-spanners, as they
require bi-paths between all pairs of vertices. However, this is not sufficient,
as shown in
Figure~\ref{fig:tct-bp-counter},
\begin{figure}[ht]
  \centering
  \begin{tikzpicture}[every loop/.style={}, thick, node distance=1.2cm]
    \tikzset{vertex/.style={draw, circle, inner sep=0.55mm,fill=black}}

    \node[vertex] (a) at (0,0) [label=left:$a$] {};
    \node[vertex] (b) at (1,1) [label=$b$] {};
    \node[vertex] (c) at (0,2) [label=left:$c$] {};
    \node[vertex] (d) at (1,3) [label=$d$] {};
    \node[vertex] (e) at (2,2) [label=right:$e$] {};
    \node[vertex] (f) at (2,0) [label=right:$f$] {};

    \tikzstyle{every node}=[font=\scriptsize,fill=white,inner sep=0.8pt]
    \begin{scope}

    \draw (a) edge node {$1,9$} (b);
    \draw (b) edge node {$2,8$} (c);
    \draw (c) edge node {$3,7$} (d);
    \draw (d) edge node {$6$} (e);
    \draw (e) edge node {$5,10$} (b);
    \draw (b) edge node {$4,11$} (f);
    \end{scope}
  \end{tikzpicture}
  \caption{A temporally connected graph with pivot vertex $(d,7)$ and bidirectional paths between all vertices, but no temporal spanning tree.}
  \label{fig:tct-bp-counter}
\end{figure}
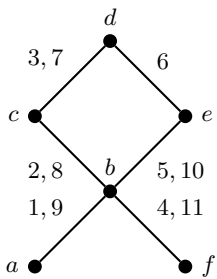
where all pairs can reach each other using
bi-paths, but these paths cannot be mutualized in the form of a
spanning tree. We will use this graph as a gadget in one of our
reductions. Observe that this graph also has a pivot vertex~$(d,7)$.

\subsection{Simple, strict, proper, happy: generality of results}

Temporal graphs can be restricted in many ways. Natural restrictions
include the fact of being \emph{simple}: every edge has a single
presence time ($\lambda$ is single-valued); and the fact of being
\emph{proper}: adjacent edges never have a label in common (that is,
every snapshot is a collection of matchings and isolated vertices).
Temporal graphs that are both simple and proper are called
\emph{happy}.

A convenient fact about properness is that it makes the distinction
between strict and non-strict vanish. Our main negative result
establishes that \TST is NP-complete even in \emph{proper} temporal
graphs, thus this is true in both the strict and non-strict settings.
Observe that \TST corresponds to \BS for the particular case that
$k=n-1$, so we also get by the same result that \BS is
NP-complete. However, we establish that \BS is also NP-complete in \emph{simple} temporal
graphs (where \TST is shown to be tractable). Since the number of
edges and the number of time labels coincide in this setting of temporal
graphs, our result establishes that minimizing the number of labels
instead of the number of edges would also be hard, a benefit of
targeting \emph{simple} temporal graphs. We insist that targeting
well-chosen settings of temporal graphs in the reductions allow one to
obtain more general (and unified) results. In fact, if the
problem were hard in \emph{happy} temporal graphs, then a single
reduction would suffice to produce all the above results. This
is not the case here, as both problems can be solved trivially in happy graphs.
The reader is referred to~\cite{CCS24} for a study of how these
various settings impact the expressivity of temporal graphs in terms
of reachability (as well as~\cite{doering2025} for directed temporal graphs).

To summarize, the tractability landscape of the problems considered in this paper is depicted in Figure~\ref{fig:outline}.

\begin{figure}[!ht]
  \centering
\begin{tikzpicture}[scale=.6]
  \tikzstyle{alwaysno} = [pattern={vertical lines},pattern color=lightgray]
  \tikzstyle{bothpoly} = [pattern={vertical lines},pattern color=lightgray]
  \tikzstyle{bispannerhard} = [pattern={north east lines},pattern color=gray]
  \tikzstyle{bothhard} = [fill=red!20]

  \tikzstyle{every node}=[font=\footnotesize]

  \draw[thick,black] (0,0) circle[radius=3cm];
  \draw[thick,black] (4,0) circle[radius=3cm];
  \node [draw,fit=(current bounding box),inner sep=5mm] (frame){};
  \begin{scope} 
  	\clip (4,0) circle[radius=3cm];
  	\draw [fill=green!20] (0,0) circle[radius=3cm];
  	\draw [alwaysno] (0,0) circle[radius=3cm];
  \end{scope}

  \begin{scope}[even odd rule] 
    \clip (4,0) circle[radius=3cm];
  	\draw [draw=none, bothhard] (4,0) circle[radius=3cm] (0,0) circle[radius=3cm];
  \end{scope}

  \begin{scope}[even odd rule] 
  	\clip (-3,-3) rectangle +(6,3);
    \clip (4,0) circle[radius=3cm]
          (4,0) circle[radius=10cm];
  	\draw [fill=orange!20] (0,0) circle[radius=3cm];
  	\draw [bispannerhard] (0,0) circle[radius=3cm];
  \end{scope}

  \begin{scope}[even odd rule] 
  	\clip (-3,0) rectangle +(6,3);
    \clip (4,0) circle[radius=3cm]
          (4,0) circle[radius=10cm];
  	\draw [fill=green!20] (0,0) circle[radius=3cm];
  	\draw [bothpoly] (0,0) circle[radius=3cm];
      \end{scope}

  \begin{scope}[even odd rule] 
    \clip (0,0) circle[radius=3cm]
          (current bounding box.north west) rectangle (current bounding box.south east);
    \clip (4,0) circle[radius=3cm]
          (current bounding box.north west) rectangle (current bounding box.south east);
  	\draw [draw=none, bothhard] (-4,-4) rectangle +(16,4);
  \end{scope}

  \begin{scope}[even odd rule] 
    \clip (0,0) circle[radius=3cm]
          (current bounding box.north west) rectangle (current bounding box.south east);
    \clip (4,0) circle[radius=3cm]
          (current bounding box.north west) rectangle (current bounding box.south east);
  	\draw [draw=none, bothhard] (-4,0) rectangle +(16,4);
  \end{scope}

  \draw[thick,black] (0,0) circle[radius=3cm];
  \draw[thick,black] (4,0) circle[radius=3cm];

  \node at (0:2cm) {Happy};
  \node at (-.3, 1) {Simple};
  \node at (4.3, 1) {Proper};
  \node at (-2, 3.1) {Strict};
  \node at (-2,-3.2) {Non-strict};

  \draw (-3.7,0) -- (-3,0);
  \draw[dashed] (-3,0) -- (1,0);
  \draw (7,0) -- (7.7,0);

  \tikzstyle{every path}=[thick]
\end{tikzpicture}
  \caption{Outline of the results. (\textit{Plain red:} both \TST and \BS are NP-hard; \textit{Vertical green:} both problems are tractable; and \textit{Slanting orange:} \TST is tractable and \BS is NP-hard. 
  )}
  \label{fig:outline}
\end{figure}

\section{\TC spanning trees}
\label{sec:tst}

In this section, we first show that all \TC trees contain a pivot contact or a pivot vertex.
This property is subsequently used in the reduction from SAT to \TST. We also observe that the problem is trivial to decide in the case of simple temporal graphs.

\subsection{The existence of pivot structures in \TC trees}
We now show that \TC trees always admit a pivot contact or a pivot vertex.
Note that the existence of such a pivot immediately implies that the temporal graph is temporally connected.
\begin{lemma}
Let~$\G$ be a temporal graph. 
If~$\G$ contains a pivot vertex, then~$\G$ is temporally connected.
If~$\G$ contains a strict (non-strict) pivot contact, then~$\G$ is strict (non-strict) temporally connected.
\end{lemma}

Moreover, for temporal trees, the converse also holds.

\begin{theorem}\label{pivot in trees}
Let~$\G$ be a temporal tree on at least two vertices.
If~$\G \in \sTC$, then $\G$ has a pivot contact or a pivot vertex.  
If~$\G \in \TC$, then it has a pivot contact.  
\end{theorem}

To show the statement, we first prove that whenever the footprint of a \TC graph has a bridge (edge whose removal disconnects it), then we can always restrict the number of labels of this edge to at most~$2$ while preserving temporal connectivity. This fact also appears as Lemma~3 in~\cite{inefficiently}, however we need the following specific proof for subsequent use.

\begin{lemma}
  \label{lem:two-labels}
  Let $\G$ be a strict (non-strict) temporally connected graph. If $uv$ is a bridge in the footprint and $|\lambda(uv)| \geq 2$, then one can restrict $\lambda(uv)$ to at most two labels~$\lambda^{\to}(uv)$ and~$\lambda^{\from}(uv)$ while preserving strict (non-strict) temporal connectivity.
\end{lemma}
\begin{proof}
  Let $V_u$ and $V_v$ be the two components after removing~$uv$ from~$\G$. For all $x\in V_u$, let $t({x},v)$ be the earliest time $x$ can reach $v$ via a strict (non-strict) temporal path.
  Note that $t(x,v)\in \lambda(uv)$.
  Now, let $\lambda^{\to}(uv) =\max \{t(x,v) \mid x\in V_u\}$.
By definition, at least one vertex of $V_u$ cannot cross $uv$ earlier than $\lambda^{\to}(uv)$, and yet, this vertex can reach all of $V_v$ since $\G$ is temporally connected. 
Thus, each vertex of~$V_u$ can traverse edge~$uv$ at time~$\lambda^{\to}(uv)$ and still reach each vertex of~$V_v$ afterwards. Symmetrically, one can identify a label $\lambda^{\from}(uv)$ such that all the vertices in $V_v$ can reach $u$ at time $\lambda^{\from}(uv)$, and reach all of $V_u$ subsequently. Thus, $\lambda^{\to}(uv)$ and $\lambda^{\from}(uv)$ are sufficient for preserving reachability between $V_u$ and $V_v$. To conclude, observe that reachability within $V_u$ or within $V_v$ does not need $uv$.
\end{proof}

Based on this, we can now show the following.
\begin{theorem}\label{pivot hereditary bridge}
Let $\G$ be a strict (non-strict) temporally connected graph that has a bridge edge $uv$ with~$\min \lambda(uv) = \lambda^{\to}(uv)\leq\lambda^{\from}(uv) = \max \lambda(uv)$.
Moreover, let~$V_v$ be the connected component of~$v$ after removing edge~$uv$ from~$\G$.
Then, (i)~if $\G[V_v\cup \{u\}]$ has a strict (non-strict) pivot contact, then~$\G$ has a strict (non-strict) pivot contact, and (ii)~if $\G[V_v\cup \{u\}]$ has a pivot vertex, then~$\G$ has a pivot vertex.
\end{theorem}
\begin{proof}
Let~$V_u$ denote the connected component of~$u$ in~$\G-uv$.
Recall that, $\lambda^{\to}(uv)$ is the smallest label of~$uv$ in~$\G$, such that each vertex of~$V_u$ can traverse the edge~$uv$.
Similarly, $\lambda^{\from}(uv)$ is the smallest label of~$uv$ in~$\G$, such that each vertex of~$V_v$ can traverse the edge~$uv$.

Now consider the temporal subgraph~$\G^* := \G[V_v\cup \{u\}]$.
We now distinguish which pivot structure~$\G^*$ contains.

\textbf{Case 1:} $\G^*$ has a pivot contact~$(e,t)$\textbf{.}
This implies that in particular there is a strict (non-strict) temporal path~$P_1$ that starts at vertex~$u$ and traverses the edge~$e$ at time~$t$, and there is a strict (non-strict) temporal path~$P_2$ that traverses the edge~$e$ at time~$t$ and ends in vertex~$u$.
We show that~$(e,t)$ is also a pivot contact of~$\G$.
As~$\lambda^\to(uv) = \min \lambda(uv)$ and $\lambda^{\from}(uv) = \max \lambda(uv)$, the departure time of~$P_1$ is at least~$\lambda^{\to}(uv)$ and the arrival time of~$P_2$ is at most~$\lambda^{\from}(uv)$.
By definition of~$\lambda^{\to}(uv)$, each vertex of~$V_u$ can traverse the edge~$uv$ at that time, which implies that we can extend the path~$P_1$ at the beginning to start from any vertex of~$V_u$ in~$\G$.
We now similarly show that we can extend the path~$P_2$ at the end to also reach any vertex of~$V_u$ in~$\G$.
As~$\G$ is temporally connected, each vertex of~$V_v$ can reach each vertex of~$V_u$.
This holds in particular for any vertex of~$V_v$ for which~$\lambda^{\from}(uv)$ is the earliest time this vertex can traverse the edge~$uv$ in~$\G$.
This implies that for each vertex~$u'$ of~$V_u$, there is a strict (non-strict) temporal path in~$\G$ that traverses the edge~$uv$ at time~$\lambda^{\from}(uv)$ and reaches~$u'$.
Thus, we can extend~$P_2$ at the end to reach any vertex of~$V_u$ in~$\G$.  
This implies that~$(e,t)$ is also a picot contact of~$\G$.

\textbf{Case 2:} $\G^*$ has a pivot vertex~$(w,t)$\textbf{.}
Similarly to the previous case, we can define strict temporal paths~$P_1$ and~$P_2$.
Then, with the same arguments about the labels~$\lambda^{\to}(uv)$ and~$\lambda^{\from}(uv)$ on~$uv$, we can show that~$(w,t)$ is a pivot vertex of~$\G$ too.
\end{proof}

This is of independent interest but in particular allows us to prove the existence of pivot structures in \TC trees.

\begin{proof}[Proof of~\Cref{pivot in trees}]
First, we show the statement for \TC stars.
Let~$\G$ have a star as underlying graph with center vertex~$c$.
Let~$e$ be any edge for which~$\min \lambda(e)$ is maximized, and let~$t:= \min \lambda(e)$.
Moreover, let~$v$ be the endpoint of~$e$ that is distinct from~$c$.
We distinguish two cases of strict and non-strict reachability.

\textbf{Case 1:} We consider non-strict temporal paths\textbf{.}
We will show that~$(e,t)$ is a non-strict pivot contact.
This is due to the following:
Firstly, since for each other vertex~$u\notin \{v,c\}$, we have~$\min \lambda(uc)\leq t$, there is a non-strict temporal path starting at~$u$ that traverses edge~$e$ at time~$t$.
Secondly, since~$\G$ is non-strictly temporally connected, for each vertex~$u\notin \{v,c\}$, there is a temporal path~$P$ from~$v$ to~$u$.
As we are dealing with a star, the first edge of~$P$ is~$e$ and the label of that edge used by~$P$ is at least~$t = \min \lambda(vc)$.
Moreover, as~$e$ is the first edge of that path, there is a temporal path~$P'$ that agrees with~$P$ on everything but possibly the label of the first edge, and the first edge is traversed at time~$t$.
Hence, for each vertex~$u\notin\{v,c\}$, there is a temporal path that traverses edge~$e$ at time~$t$ and arrives at~$u$.
Consequently, $(e,t)$ is a pivot contact.

\textbf{Case 2:} We consider strict temporal paths\textbf{.}
We will show that~$(e,t)$ is a strict pivot contact or that~$(c,t+1)$ is a strict pivot vertex.
Assume that~$(e,t)$ is not a strict pivot contact, as otherwise, we are done.
Consider the arguments used in the previous case. 
The arguments for the existence of the temporal paths that traverse~$e$ at time~$t$ and reach any vertex~$u$ still holds in the strict case.
Hence, the only thing that could prevent~$(e,t)$ from being a strict pivot contact is the possibility that there is some vertex~$u\notin\{v,c\}$ such that there is no strict temporal path starting at~$u$ that traverses edge~$e$ at time~$t$.
Since we are dealing with a star, that is, since~$u$ is a neighbor of~$c$, this implies that~$\min \lambda(uc) \geq t = \min \lambda(vc)$.
By choice of~$e=vc$, we have~$t \geq \min \lambda(uc)$, which gives us~$\min \lambda(uc) = t$.
Due to the temporal connectivity of~$\G$, this implies that for each vertex~$w\neq c$, the edge~$wc$ has a label which is at least~$t+1$.
Based on this fact, we now show that~$(c,t+1)$ is a pivot vertex.
First, by choice of~$t$ and the fact that~$\G$ has a star as underlying graph, we immediately get that for each vertex~$w$ of~$\G$,  there is a temporal path from~$w$ that arrives at~$c$ with an edge of label at most~$t < t+1$.
It thus remains to show that there is also a temporal path starting at time at least~$t+1$ from~$c$.
As already shown, for each vertex~$w\neq c$, the edge~$wc$ has a label which is at least~$t+1$.
Hence, also such temporal paths exists, which implies that~$(c,t+1)$ is a pivot vertex.

Thus, the statement holds for stars.

Next, we show that the statement holds for all \TC trees via induction on the number~$n$ of vertices. 

For the base case, consider~$n\in \{2,3\}$.
For~$n=2$, the statement trivially holds, as each temporal edge of the graph connects all vertices and thus is a pivot contact.
For~$n=3$, the statement holds, as each tree on three vertices is a star.
This completes the base case.

For the inductive step, let~$\G$ be a temporal tree on~$n\geq 4$ vertices that is strictly (non-strictly) temporally connected and assume that the statement holds for all temporal trees that are strictly (non-strictly) temporally connected on at most~$n-1$ vertices.
Let~$G$ be the underlying graph of~$\G$.
If~$G$ is a star, the statement holds by the initial proof.
So assume that~$G$ is not a star.
Then, there is some edge~$uv$ in~$G$, such that~$G - \{uv\}$ consists of two trees~$T_u$ and~$T_v$ with at least two vertices each, such that~$u$ is in~$T_u$ and~$v$ is in~$T_v$.
By Lemma~\ref{lem:two-labels}, we can restrict the labels on~$uv$ to~$\{\lambda^{\to}(uv),\lambda^{\from}(uv)\}$ while preserving that the resulting temporal tree~$\G'$ is strictly (non-strictly) temporally connected.
Here, $\lambda^{\to}(uv)$ is the smallest label of~$uv$ in~$\G$, such that each vertex of~$T_u$ can traverse the edge~$uv$.
Similarly, $\lambda^{\from}(uv)$ is the smallest label of~$uv$ in~$\G$, such that each vertex of~$T_v$ can traverse the edge~$uv$.
Assume without loss of generality that~$\lambda^{\to}(uv) \leq \lambda^{\from}(uv)$
As~$\G'$ is strict (non-strict) temporally connected, so is~$\G^*$.
Moreover, $\G^*$ contains strictly less than~$n$ vertices, as~$T_u$ has at least two vertices, and we only kept the vertex~$u$ of that subtree.
By the induction hypothesis, this implies that~$\G^*$ contains either
\begin{itemize}
\item a pivot contact or a pivot vertex, if~$\G^*$ is strictly temporally connected
\item a pivot contact, if~$\G^*$ is non-strictly temporally connected.
\end{itemize}
Hence, the conditions of~\Cref{pivot hereditary bridge} hold for bride edge~$uv$, which implies that~$\G'$ (and thus also~$\G$) contains the respective pivot structures too.
\end{proof}

\subsection{NP-completeness}

The fact that \TST is in NP is straightforward: given an input graph $\G$ and a candidate subgraph $\G'$, it is easy to verify that $\G'$ is a temporal spanner of~$\G$ and the footprint of $\G'$ is a tree.
The exact complexity of these steps may depend on the data structure used for encoding the temporal graph. However, for all reasonable data structures (including the ones discussed in Section~\ref{sec:definitions}), each step is clearly polynomial, using for example any of the algorithms from~\cite{BFJ03} for testing temporal connectivity.

 We now prove that \TST is NP-hard in \emph{proper} temporal graphs (with the consequences discussed above), reducing from \textsc{SAT}.
In the reduction we will use the fact that the solution must have a pivot.

\begin{theorem}
  \label{th:tst-hard degree}
  \TST is NP-hard in proper temporal graphs, where the underlying graph has a maximum degree of~$5$, and even if every edge has at most 2 time labels.
Even under these restrictions, \TST cannot be solved in $2^{o(n+m)}$~time without violating the  ETH.
\end{theorem}
\begin{proof}
  Let $\phi$ be a CNF formula with $n$ variables and $m$ clauses, we will transform $\phi$ into a proper temporal graph $\G$ such that $\G$ admits a \TC spanning tree if and only if $\phi$ is satisfiable. 
  Without loss of generality, we assume that none of the clauses contain both a variable and its negation. 
To obtain the desired degree bound on the the underlying graph, we further assume that each clause contains at most~$3$ literals and that each variable (summed up over the positive and negative literals) occurs in exactly 3 clauses.
  For simplicity, we will use both negative and positive time labels, and chose some time labels as rational numbers. 
  A suitable renormalization can be applied easily to convert this graph into a graph with time labels from~$\mathbb{N}$ while preserving the relative order of the labels.

  The graph $\G$ has for each variable~$x_i$ the vertices $x_i$, $F_i$, and~$T_i$, a vertex $c_j$ for each clause~$c_j$, and two special vertices (thus, $3n+m+2$ vertices in total). The special vertices are $B$ and~$B'$. 
   Intuitively, we will ensure that the edge~$BB'$ will be the only pivot in $\G$, and thus also a pivot in each \TC spanning tree of~$\G$. 
   The edge of the footprint are built as follows. 
   There is an edge between $B$ and $T_1$, and an edge between $B$ and $F_1$, an edge between~$B$ and~$B'$.
   Further, there is an edge between~$x_1$ and~$T_1$ and an edge between~$x_1$ and~$F_1$. 
   Then, for each variable $x_i$ with~$i > 1$, there is an edge between $T_i$ and $x_i$, and an edge between $F_i$ and $x_i$, an edge between~$T_i$ and~$T_{i-1}$, and an edge between~$F_i$ and~$F_{i-1}$. 
   Finally, for each clause $c_j$ where $x_i$ appears (positively or negatively), there is an edge between $x_i$ and $c_j$.
  Parts of the footprint of $\G$ are illustrated in Figure~\Cref{fig:new hardness}.
  
  The time labels are as follows (see also Figure~\Cref{fig:new hardness}).
  Edge~$BB'$ is given the single label~$0$.
  Every other edge is given two time labels, one positive (defined relative to a reference value $t^+ = n+k+1$), and one negative (defined relative to $t^- = -t^+$). 
  Let~$\epsilon = \frac{1}{n+1}$.
  The time labels on the path~$P_F:=(B,F_1, \dots, F_n)$ are defined, such that (i)~starting at vertex~$F_n$ at time~$-n \cdot \epsilon = -\frac{n}{n+1}$, one can reach vertex~$B$ at time~$-\epsilon$ while traversing only edges of~$P_F$ and (ii)~starting at vertex~$B$ at time~$t^++\epsilon$, one can reach vertex~$F_n$ at time~$t^+ + n\cdot\epsilon$ while traversing only edges of~$P_F$.
Formally, for each~$i\in [1,n]$, the edge~$F_iF_{i-1}$ receives the labels~$-i\cdot \epsilon$ and~$t^+ +  i\cdot \epsilon$, where for simplicity, $F_0 := B$.
  Analogously, the time labels on the path~$P_T:=(B,T_1, \dots, T_n)$ are defined, such that (i)~starting at vertex~$T_n$ at time~$t^-+\epsilon$, one can reach vertex~$B$ at time~$t^- + n\cdot \epsilon < t^-+1$ while traversing only edges of~$P_T$ and (ii)~starting at vertex~$B$ at time~$\epsilon$, one can reach vertex~$T_n$ at time~$n\cdot\epsilon$ while traversing only edges of~$P_T$.
Formally, for each~$i\in [1,n]$, the edge~$T_iT_{i-1}$ receives the labels~$t^-+(n-i+1)\cdot \epsilon$ and~$ i\cdot \epsilon$, where for simplicity, $T_0 := B$.
   For each variable $x_i$, the labels on $T_i x_i$ are $t^- - i$ and $i$; the ones on $F_i x_i$ are $-i$ and $t^+ + i$. 
   Finally, for each clause $c_j$ where $x_i$ appears, we assign time labels to the edge $x_ic_j$ depending on whether $x_i$ appears positively or negatively in $c_j$, namely:
  \begin{itemize}
    \item If $x_i \in c_j$, the labels are $t^- - (i + j)$ and $(i + j)$. In this case, $x_i c_j$ is called a \emph{positive edge}.
    \item If $\neg x_i \in c_j$, the labels are $- (i + j)$ and $t^+ + (i + j)$. In this case, $x_i c_j$ is called a \emph{negative edge}.
  \end{itemize}

    \begin{figure}[htb]
          \centering
          \begin{tikzpicture}[every loop/.style={}, thick, node distance=1.2cm,auto,xscale = 2,yscale=1.2]
              \tikzset{vertex/.style={draw, circle, inner sep=0.55mm,fill=black}}
              \tikzstyle{tx} = [line width=2pt,black!50!green]
              \tikzstyle{positive} = [dashdotted, line width=2pt,black!50!green]
              \tikzstyle{fx} = [line width=2pt,red]
              \tikzstyle{negative} = [dotted, line width=2pt,red]

              \node[vertex] (a) at (1,-3) [label=left:$c_j$] {};
              \node[vertex] (b) at (1,-1.5) [label=$x_i$] {};
              \node[vertex] (c) at (0,2) [label=left:$T_1$] {};
              \node[vertex] (d) at (1,3) [label=below:$B$] {};
              \node[vertex] (d2) at (1,4) [label=above:$B'$] {};
              \node[vertex] (e) at (2,2) [label=right:$F_1$] {};
              
              \node[vertex] (c2) at (0,-.5) [label=left:$T_i$] {};              
              \node[vertex] (e2) at (2,-.5) [label=right:$F_i$] {};
              \node[vertex] (c3) at (0,.5) [label=left:$T_{i-1}$] {};              
              \node[vertex] (e3) at (2,.5) [label=right:$F_{i-1}$] {};

              \begin{scope}

              \draw[thick, dotted] (c2) edge node[text=black] {} (c);
              \draw[thick, dotted] (e2) edge node[text=black] {} (e);
              \draw[thick, dotted] (c2) edge node[text=black] {} ($(c2) + (0,-1.3)$);
              \draw[thick, dotted] (e2) edge node[text=black] {} ($(e2) + (0,-1.3)$);
              
                \tikzstyle{every node}=[black,font=\footnotesize,inner sep=1pt]
              \draw (a) edge node[right] (acc) {$\begin{cases}
              t^--(i+j),(i+j)& \textrm{if~} x_i \in c_j\\ -(i+j),t^++(i+j)& \textrm{if~} \neg x_i \in c_j
              \end{cases}$} (b);
              \draw (b) edge[tx] node[text=black] {$t^--i,i$} (c2);
              \draw (c) edge node {$t^-+n\cdot \epsilon,\epsilon$} (d);
              \draw (d) edge node {$-\epsilon,t^++\epsilon$} (e);
              \draw (d2) edge node {$~0$} (d);
              \draw (c2) edge node {$t^-+(n-i+1)\cdot \epsilon,i\cdot \epsilon~$} (c3);
              \draw (e3) edge node {$~-i\cdot \epsilon,t^++i\cdot \epsilon$} (e2);
              \draw (e2) edge[fx] node[text=black] {$-i,t^++i$} (b);
              \end{scope}
          \end{tikzpicture}
      \label{fig:new hardness}
      \caption{An illustration of the reduction from SAT to \TST. Here, clause clause~$c_j$ contains a literal of variable~$x_i$, and the edge~$c_jx_i$ receives labels according to which literal of~$x_i$ is in~$c_j$.
      If~$x_i$ occurs positively in~$c_j$, then to preserve temporal paths between~$c_j$ and~$B'$, it suffices to keep the green edge.
      Otherwise, it suffices to keep the red edge.
      Intuitively, temporal paths with purely negative labels go around the illustration in clockwise order, whereas temporal paths with purely positive labels go around the illustration in counter-clockwise order.}
  \end{figure}

  The idea of the construction is that a spanning tree is equivalent to a valuation that satisfies $\phi$.
   A spanning tree will only allow for each variable~$x_i$ to keep at most one of the edges $T_ix_i$ or~$F_ix_i$.
   
   We show the first direction of the proof.
   \begin{nclaim}
    \label{claim:sat-tct}
    If $\phi$ is satisfiable, then $\G$ admits a \TC spanning tree.
  \end{nclaim}
  \begin{proof}
    Consider a satisfying assignment of $\phi$. 
    We describe which edges to take into the \TC spanning tree~$\G'$.
    We take the edge~$BB'$ and all edges of the paths~$(B,T_1,\dots,T_n)$ and~$(B,F_1,\dots,F_n)$.
    For each variable $x_i$, if $x_i$ is set to \texttt{true}, we take the edge~$T_ix_i$. 
    If $x$ is set to \texttt{false}, we take instead the edge~$F_ix_i$.
    Finally for every clause $c$, pick one variable $x$ which satisfies clause~$c$ under this truth assignment.
    Such a variable exists, as the truth assignment satisfies each clause. 
    
    Note that the resulting temporal graph~$\G'$ has indeed a tree as underlying graph.
    Moreover, note that some variables may have no clause vertices as neighbors in~$\G'$.
    
It thus remains to show that~$\G'$ is temporally connected.
      In order to prove that, we will prove that $BB'$ is a pivot at time $0$; in other words,
      we need to show that every vertex $u\neq B'$ can reach $B$ at a time smaller than~$0$ and can be reached by $B$ starting at time larger than $0$. 
      
      \textbf{Case 1.}      If $u$ is a clause vertex $c_j$, we let~$x_i$ be the unique neighbor of~$c_j$ in~$\G'$.
      We distinguish whether~$T_ix_i$ is an edge of~$\G'$ or~$F_ix_i$ is an edge of~$\G'$.
If~$T_ix_i$ is an edge of~$\G'$, then variable~$x_i$ is set to~\texttt{true}, which implies that~$x_i$ occurs positively in~$c_j$, based on the choice to take the edge~$c_jx_i$ into~$\G'$.
Hence, the labels on~$c_jx_i$ are~$t^--(i+j)$ and~$(i+j)$, and the labels on~$T_ix_i$ are~$t^--i$ and~$i$.
Thus, in~$\G'$, $c_j$ can reach vertex~$T_i$ prior to time~$t^--i$ and~$T_i$ can reach~$c_j$ when starting at a time~$i$.
By definition, the path~$(T_i,\dots,T_1,B)$ is in~$\G'$ and  (i)~allows~$T_i$ starting at a time strictly between~$t^-$ and~$t^- +1 > t^- + n\cdot \epsilon$ to reach~$B$ at a time smaller than~$0$ and (ii)~allows~$B$ starting at a time larger than~$0$ to reach~$T_i$ at a time strictly between~$0$ and~$1>n\cdot \epsilon$.
Hence, this implies that~$c_j$ can reach~$B$ prior to time~$0$ and~$B$ can reach~$c_j$ when starting at a time larger than~$0$.  
The case for the edge~$F_ix_i$ being in~$\G'$ can be shown symmetrically.

      \textbf{Case 2:} If~$u$ is a variable~$x_i$, then the path~$(x_i,F_i,\dots,F_1,B)$ or the path $(x_i,T_i,\dots,T_1,B)$ (depending on which of the two edges incident with~$x_i$ is in~$\G'$) are bi-paths that allow~$x_i$ to reach~$B$ prior to time~$0$ and over which~$B$ can reach~$x_i$ when starting at a time later than~$0$.
      
      \textbf{Case 3:} If~$u$ is a vertex of~$\{F_i,T_i\mid i\in [1,n]\}$, then by definition of the labels on the paths~$(F_n,\dots,F_1,B)$ and~$(T_n,\dots,T_1,B)$, the respective subpath between~$u$ and~$B$ allow~$u$ to reach~$B$ prior to time~$0$ and allow~$B$ to reach~$u$ when starting at a time later than~$0$.
      
      Thus, $BB'$ is a pivot edge of~$\G'$, which implies that $\G'$ is a \TC spanning tree of~$\G$.   \lncsqed   
        \end{proof}

In the remainder of the proof, we show that the inverse is also true.
Assume that there is \TC spanning tree~$\G'$ for~$\G$.
We show that this implies that~$\phi$ is satisfiable.
To this end, we first observe some properties about~$\G'$.
As~$BB'$ is the only edge incident with~$B'$, this edge is part of each temporal spanner of~$\G$.
Moreover, since this edge receives only the single label (namely~$0$), $BB'$ is a pivot edge in~$\G'$.
We exploit this fact as follows:
As each other edge has only one positive and one negative label, for each path~$P=(v_1,v_2,\dots,v_r=B)$ in~$\G'$ the negative labels are strictly increasing and the positive labels are strictly decreasing, as otherwise, $v_1$ cannot reach~$B$ prior to time~$0$ or~$B$ cannot reach~$v_1$ when starting at a time strictly larger than~$0$.
That is, for each~$i\in[1,r-2]$, 
\begin{itemize}
\item $\min \lambda(v_iv_{i+1}) < \min \lambda(v_{i+1}v_{i+2})$ and
\item $\max \lambda(v_iv_{i+1}) > \max \lambda(v_{i+1}v_{i+2})$.
\end{itemize}

   We now show that for each variable~$x_i$, each \TC spanning tree of~$\G$ contains at most one of $T_ix_i$ or~$F_ix_i$.

  \begin{nclaim}
    \label{claim:fixed-cycle}
    The \TC spanning tree~$\G'$ contains all edges of the two paths $(T_n,\dots,T_1,B)$ and~$(F_n,\dots,F_1,B)$.
  \end{nclaim}
  \begin{proof}
  For simplicity, let~$T_0 := F_0:= B$.
  
  Let~$i\in[1,n]$ and assume that the edge~$T_iT_{i-1}$ is not part of~$\G'$.
  We show that there is no path starting at vertex~$B$ at any time larger than~$0$ that reaches~$T_i$ in~$\G'$.
  This then contradicts the fact that~$BB'$ is a pivot edge at time~$0$ and implies that edge $T_iT_{i-1}$ has to be part of~$\G'$.
  Since we only consider paths that start at time steps larger than~$0$, we only need to consider the positive labels on the edges and will argue that there is no temporal $(B,T_i)$-path in the respective temporal graph.
  Such a path does not exist, since (i)~the positive label on~$T_ix_i$ is by definition~$i$ and thus smaller than any other positive label incident with~$x_i$ and (ii)~the positive label of~$T_iT_{i+1}$ is smaller than 1, which is strictly smaller than the smallest positive label incident with any variable vertex~$x_j$. 
    The first part implies that no such path can traverse edge~$T_ix_i$ and the second part implies that no such path can traverse the edge~$T_iT_{i+1}$ (if~$i < n$), as this edge needs to be preceded by some edge incident with a variable vertex.
    As these are the (at most) two neighbors of~$T_i$ besides~$T_{i-1}$ in~$\G$, there is in fact no temporal $(B,T_i)$-path in the respective temporal graph, if~$T_iT_{i-1}$ is not part of~$\G'$.
    Consequently, $T_iT_{i-1}$ is part of~$\G'$.
  
Now assume that the edge~$F_iF_{i-1}$ is not part of~$\G'$.
  We show that there is no path from~$F_i$ that reaches~$B$ prior to time~$0$ in~$\G'$.
  This then contradicts the fact that~$BB'$ is a pivot edge at time~$0$ and implies that edge $F_iF_{i-1}$ has to be part of~$\G'$.
  Since we only consider paths that reach~$B$ prior to time~$0$, we only need to consider the negative labels on the edges and will argue that there is no temporal $(F_i,B)$-path in the respective temporal graph.
  Such a path does not exist, since (i)~the negative label on~$F_ix_i$ is by definition~$-i$ and thus larger than any other negative label incident with~$x_i$ and (ii)~the negative label of~$F_iF_{i+1}$ is larger than -1, which is strictly larger than the largest negative label incident with any variable vertex~$x_j$. 
    The first part implies that no such path can traverse edge~$F_ix_i$ and the second part implies that no such path can traverse the edge~$F_iF_{i+1}$ (if~$i < n$), as this edge needs to be (not necessarily immediately) followed by some edge incident with a variable vertex.
    As these are the (at most) two neighbors of~$F_i$ besides~$F_{i-1}$ in~$\G$, there is in fact no temporal $(F_i,B)$-path in the respective temporal graph, if~$F_iF_{i-1}$ is not part of~$\G'$.
    Consequently, $F_iF_{i-1}$ is part of~$\G'$. \lncsqed
  \end{proof}

  Claim~\ref{claim:fixed-cycle} implies that the \TC spanning tree~$\G'$ contains for each variable~$x_i$, at most one of the edges $T_i x_i$ or $F_i x_i$, as otherwise, there would be a cycle in~$\G'$. 
  Morally, this encodes the choice of an assignment for variable $x_i$ in $\phi$. 
  Based on this property, we now prove that~$\phi$ is satisfiable.
  
  \begin{nclaim}
    \label{claim:tct-sat}
    If $\mathcal{G}$ admits a \TC spanning tree, then $\phi$ is satisfiable.
  \end{nclaim}
  \begin{proof}
  We define a truth assignment as follows:
  For each variable~$x_i$, let~$e_i$ denote the first edge of the unique path from~$x_i$ to~$B$ in~$\G'$.
  We set~$x_i$ to~\texttt{true} if the largest label of~$e_i$ is smaller than~$t^+$.
  Otherwise, we set~$x_i$ to~\texttt{false}.
  Note that by the definition of the labels on edges incident with~$x_i$, this implies:
\begin{itemize}
\item If~$x_i$ is set to~\texttt{true}, $\min \lambda(e_i) < t^-$ and~$\max \lambda(e_i) < t^+$.
\item If~$x_i$ is set to~\texttt{false}, $\min \lambda(e_i) > t^-$ and~$\max \lambda(e_i) > t^+$. 
\end{itemize}  
  We show that this truth assignment satisfies~$\phi$.
  Let~$c_j$ be an arbitrary clause of~$\phi$.
  It suffices to show that~$c_j$ is satisfied by the truth assignment.
  As each neighbor of vertex~$c_j$ in~$\G$ is a variable vertex, there is some variable~$x_i$, such that~$x_i$ is the first internal vertex of the unique path from~$c_j$ to~$B$ in~$\G'$.
Thus, $c_jx_i$ and~$e_i$ are the first two edges of the unique path~$P_j$ from~$c_j$ to~$B$.  
  As discussed initially, this implies that~$\min \lambda(c_jx_i) < \min \lambda(e_i)$ and $\max \lambda(c_jx_i) > \max \lambda(e_i)$.
  Hence, if~$x_i$ is set to~\texttt{true}, then the smaller label on~$c_jx_i$ is smaller than~$\min \lambda(e_i) < t^-$, implying that~$c_jx_i$ is a positive edge, that is, $x_i$ occurs positively in~$c_j$.
  Otherwise, if~$x_i$ is set to~\texttt{false}, then the larger label on~$c_jx_i$ is larger than~$\max \lambda(e_i) > t^-$, implying that~$c_jx_i$ is a negative edge, that is, $x_i$ occurs negatively in~$c_j$.
  In both cases, the truth assignment satisfies clause~$c_j$ via the literal of~$x_i$ that occurs in~$c_j$.
  As a consequence, the defined truth assignment satisfies each clause, which implies that~$\phi$ is satisfiable.\lncsqed
  \end{proof}

  From Claims~\ref{claim:sat-tct} and~\ref{claim:tct-sat}, $\phi$ is
  satisfiable if and only if $\G$ admits a \TC spanning tree.
  Furthermore, $\mathcal{G}$ is proper and each edge has at most $2$ time labels. 
  As the number of vertices and the number of time edges in~$\G$ is linear in the size of $\phi$, an algorithm for~\TST with running time~$2^{o(n+m)}$ would imply an algorithm for~SAT with running time~$2^{o(\phi)}$. 
  The latter would violate the ETH.
  This completes the proof.\lncsqed
\end{proof}

We now argue that we can also transfer the hardness to instances where the underlying graph is bipartite.

\begin{lemma}
\TST is NP-hard in proper temporal graphs, where the underlying graph is bipartite, has a maximum degree of~$5$, and if every edge has at most 2 time labels.
Even under these restrictions, \TST cannot be solved in $2^{o(n+m)}$~time without violating the ETH.
\end{lemma}  
\begin{proof}
Consider the instance~$\G$ of \TST obtained by the previous reduction.
Subdivide each such edge~$X_iX_{i+1}$ (with~$X\in \{T,F\}$) with labels~$\{\ell,h\}$ by a new vertex~$X'_{i+1}$ and assign the labels~$\{\ell+\mu,h-\mu\}$ to~$X_iX'_{i+1}$ and the labels~$\{\ell-\mu,h+\mu\}$ to~$X'_{i+1}X_{i+1}$ for~$\mu = \frac{\epsilon}{2}$.
Let~$\G'$ be the resulting temporal graph. 
As subdivision does not increase the maximum degree, the underlying graph of~$\G'$ has a maximum degree of~$5$.
Moreover, the underlying graph of~$\G'$ is bipartite.
To see this, let~$U:=\{T_i,F_i\mid 0\leq i \leq n\}$ and~$U':= \{B\}\cup \{T_i',F_i'\mid 1\leq i \leq n\}$, and observe that 
\begin{itemize}
\item each clause vertex is only adjacent to variable vertices,
\item each variable vertices is only adjacent to clause vertices and vertices of~$U$,
\item each vertex of~$U$ is only adjacent to variable vertices and vertices of~$U'$,
\item vertices of~$U'$ are only adjacent to vertices of~$U$.
\end{itemize}
That is, the variable vertices together with the vertices of~$U'$ are one side of the bipartition and the vertices of~$U$ together with the clause vertices form the other side.
 
We now argue that~$\G$ and~$\G'$ are equivalent instances.
Recall that in each \TC spanning tree of~$\G$, each edge of~$\{T_iT_{i+1},F_iF_{i+1}\mid 0\leq i \leq n-1\}$ is contained (see Claim~\ref{claim:fixed-cycle}).
Thus, for each edge~$X_iX_{i+1}$ in~$\G$, each \TC spanning tree of~$\G'$ has to contain both edges~$X_iX'_{i+1}$ and~$X'_{i+1}X_{i+1}$.
The remainder of the proof is then completely identical to the one of the previous reduction.
Consequently, the statement follows.\lncsqed
\end{proof}

The following lemma also implies that hardness of the problem also transfers on underlying graphs that are supergraphs of hard instances.

\begin{lemma}
Let~$\G$ be an instance of~\TST with at least three vertices.
Moreover, let~$e$ be an edge that receives no label in~$\G$ and let~$h$ be the highest label assigned to any edge by~$\G$.
Then, adding edge~$e$ to~$\G$ at time~$h+1$ yields an equivalent instance of~\TST.
\end{lemma}
\begin{proof}
Let~$\G^*$ be the resulting instance of~\TST.
We show that in fact, $\G$ and~$\G^*$ share the same \TC spanning trees.
As~$\G$ is a temporal subgraph of~$\G^*$, each \TC spanning tree of~$\G$ is also a \TC spanning tree of~$\G^*$.
Now consider the opposite direction.
We show that each \TC spanning tree~$\G'$ of~$\G^*$ is in fact also a \TC spanning tree of~$\G$.
That is, we show that~$\G'$ does not contain the edge~$e$.
Assume towards a contradiction that this is not the case.
Let~$ab := e$.
Since~$\G'$ is a tree, there are two connected components~$A$ and~$B$ after removing edge~$e$, where~$a\in A$ and~$b\in B$.
As~$\G'$ has at least three vertices, we can assume without loss of generality that~$a$ has at least one neighbor~$c$ in~$A$.
Thus, $ca$ and~$ab$ are adjacent edges in~$\G'$.
By definition, the unique label of~$ab$ is~$h+1$ and the largest label of~$ca$ is at most~$h$.
Since~$\G'$ is a tree, this implies that there is no temporal path from~$b$ to~$c$ in~$\G'$.
This contradicts the assumption that~$\G'$ is a \TC spanning tree, as~$\G'$ is not temporally connected.
As a consequence, each \TC spanning tree of~$\G^*$ uses only edges of~$\G$, and is thus a \TC spanning tree of~$\G$.\lncsqed
\end{proof}

This implies the following by repeatedly adding edges to the temporal graph.

\begin{corollary}
\TST is NP-hard even on proper temporal graphs, where the underlying graph is a clique, and each edge receives at most two labels.
Even under these restrictions, \TST cannot be solved in $2^{o(n)}$~time without violating the ETH.
\end{corollary}

\subsection{Tractability in simple temporal graphs}
\label{sec:simple-tst}

Here, we observe that \TC spanning trees do not exist, at all, in simple temporal graphs in the strict setting. In the non-strict setting, they may or may not exist, and this can be decided efficiently.

\begin{lemma}
  \label{lem:never-exists}
  Simple \TC graphs on at least three vertices in the strict setting do not admit \TC spanning trees.
\end{lemma}
\begin{proof}
  Recall that a \TC spanning tree must offer bidirectional temporal paths (bi-paths) between all pairs of vertices. If the labeling is simple, then every edge has only one presence time, but since the strict setting imposes that the labels increase along a path, this implies that no path of length $2$ or more can be used in both directions.\lncsqed
\end{proof}

\begin{lemma}
  \TST can be solved in polynomial time in simple temporal graphs in the non-strict setting.
\end{lemma}

\begin{proof}
  By the same arguments as above, no path in the tree may rely on time labels that increase at some point along the path, as this would prevent bidirectionality (the temporal graph being simple). However, in the non-strict case, repetition of the time labels are allowed along a path. Thus, a \TC spanning tree exists if and only if a (standard) spanning tree exists in at least one of the snapshots, which can be tested (and found, if applicable) in polynomial time. \lncsqed
\end{proof}

This completes the tractability landscape of \TST in the considered settings, as summarized by Figure~\ref{fig:outline}.

\section{Finding bi-paths and bi-spanners of unrestricted size}
\label{sec:bs}

In this section, we present a polynomial time algorithm that tests if a given temporal graph admits a bidirectional temporal spanner (a bi-spanner). The algorithm relies on a more basic primitive that finds, for a given node, all the bidirectional temporal paths (bi-paths) between this nodes and other nodes. If this algorithm, executed from each node, always reaches all the other nodes, then the union of the corresponding bi-paths forms a bi-spanner.
Bi-spanners obtained in this way may be larger than needed. However, there exist pathological cases where no smaller bi-spanners exist, making this algorithm worst-case optimal (admittedly, these cases are extreme, which motivates the search for small solutions in general, as discussed in Section~\ref{sec:kbs}).

\subsection{High-level description}

Let $s$ be the source node.
The algorithm consists of updating recursively, for each node $v$, a set of triplets $B_v^s=\{(u_i,a_i, d_i)\}$, where $u_i\in N(v)$ and $a_i, d_i \in [1,\tau]$, such that there exists a bi-path between $s$ and $v$ that
\begin{itemize}
\item starts at~$s$ and arrives at time at most $a_i$ at $v$ through its neighbor $u_i$ and
\item arrives at~$s$ after it departs from $v$ through $u_i$ at time at least $d_i$.
\end{itemize}
Note that $a_i$ and $d_i$ do not need to satisfy a particular relative order, as both direction of the bi-path are independent in time. The triplets are updated according to the following extension rule (see also Figure~\ref{fig:extension}).

\begin{erule}[Extension rule]
  Let $(u_i,a_i, d_i)$ be a triplet at node $v$, let $w$ be a neighbor of $v\ne u_i$. Let $a_i'=\min \{t\in \lambda(vw) \mid t \ge a_i\}$ and $d_i'=\max \{t\in \lambda(vw) \mid t \le d_i\}$, or $\bot$ if the resulting set is empty. If $a_i'\ne \bot$ and $d_i'\ne \bot$, then $(v,a_i',d_i')$ is a triplet at node~$w$.
\end{erule}

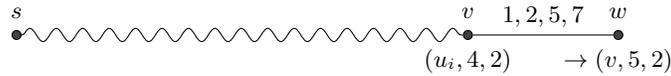
\begin{figure}[h]
  \centering
  \begin{tikzpicture}[xscale=2]
  \tikzstyle{every node}=[draw,circle,black,fill=darkgray,inner sep=1.2pt]
  \path (0,0) node (u){};
  \path (3,0) node (v){};
  \path (4,0) node (w){};
  \tikzstyle{every node}=[font=\small,above=3pt]
  \path (u) node {$s$};
  \path (v) node {$v$};
  \path (w) node {$w$};

  \draw[decorate,decoration=snake] (u) -- (v);
  \tikzstyle{every node}=[font=\footnotesize]
  \draw (v)-- node[above]{$1,2,5,7$} (w);
  \path (v) node[below=2pt] (v2){$(u_i,4,2)$};
  \path (w) node[below=2pt] (v2){$\to (v,5,2)$};
\end{tikzpicture}
\caption{\label{fig:extension} Example of the extension of a bi-path.}
\end{figure}

\begin{lemma}
  \label{lem:extension}
  If the triplet $(u_i,a_i,d_i)$ corresponds to a valid bi-path between $s$ and $v$, and a triplet $(v,a_i',d_i')$ is added by the extension rule to a neighbor $w$ of $v$, then $(v,a_i',d_i')$ corresponds to a valid bi-path between $s$ and $w$.
\end{lemma}
\begin{proof}
  The existing triplet at $v$ implies that a bi-path exists from $s$ to $v$ arriving at $v$ before time $a_i$ (included). If $a_i'\ne \bot$, then the extension rule guaranties that $a_i' \ge a_i$ and $a_i' \in \lambda(vw)$, thus the corresponding journey from $s$ to $v$ in this bi-path can be extended to $w$ through edge $vw$. Likewise, if $d_i' \ne \bot$, then $d_i' \le d_i$ and $d_i'\in \lambda(vw)$, so the journey from $v$ to $s$ can be preceded by a 1-hop journey from $w$ to $v$. \lncsqed
\end{proof}

\begin{remark}
  It is sufficient to replace $\ge$ with $>$ and $\le$ with $<$ in the extension rule for considering strict journeys instead of non-strict journeys. 
\end{remark}

Clearly, some triplets are strictly better than others. Namely, if a bi-path from the source arrives at an earliest time and departs back to it at a later time compared to second bi-path \emph{through the same neighbor}, then all journeys using the corresponding times could be replaced by solution using the times of the first bi-path, plus extra waiting at the neighbor. Thus, it is useless to maintain the second triplet in the set. This is formalized through the following
rule:

\begin{erule}[Elimination rule]
  Let $(u_i,a_i, d_i)$ and $(u_i',a_i', d_i')$, with $u_i=u'_i$, be two triplets at node $v$, if $a_i \leq a_i'$ and $d_i \geq d_i'$, then $(u_i,a_i',d_i')$ is removed from $B_v^s$.
\end{erule}

In the algorithm, we denote by $+$ the operation that consists of adding a new triplet to a set of triplets while taking into account this elimination rule. Then, we say that a node $v$ has been \emph{improved} by the addition of a triplet $b$ if $B_v^s + b \ne B_v^s$. Observe that, despite the elimination rule, the set of triplets at a node may contain several triplets involving the same neighbor, if they are incomparable in time. Similarly, several triplets may co-exist with equal times if they involve different neighbors.

\subsection{The algorithm}

The algorithm is quite compact, though its correctness and complexity are not immediate due to the fact that an edge (and even a contact on that edge) may be involved several times in the improvement of its endpoints.
Initially, $B_v^s=\emptyset$ for all $v\ne s$ and $B_s^s=\{((\bot,-\infty,\infty))\}$. The algorithm starts with calling \texttt{improveNeighbors($s$)}. Then, it improves the neighbors recursively until all the bi-paths from $s$ have been computed.

\SetKwComment{Comment}{/* }{ */}
\begin{algorithm}[h]
\caption{Procedure \texttt{improveNeighbors()} that computes all the triplets from a source.}\label{alg:bidirectional}
\textbf{Input}: current node $u$\\
  \ForAll{$($-$, a_i, d_i) \in B_u^s$}{
    \ForAll{$v \in N(u)$}{
      $a_i' \gets \min \{t\in \lambda(uv) \mid t \ge a_i\}$\\
      $d_i' \gets \max \{t\in \lambda(uv) \mid t \le d_i\}$\\
      \If {$B_v^s + (u, a_i', d_i') \ne B_v^s$}{
        $B_v^s \gets B_v^s + (u, a_i', d_i')$\\
        \textup{\texttt{improveNeighbors($v$)}}
      }
    }
  }
\end{algorithm}

\begin{theorem}
  Algorithm~\ref{alg:bidirectional} computes all bi-paths between $s$ and the other vertices.
\end{theorem}

\begin{proof}
  First, let us prove by induction that at any step, for any vertex $v$, the triplets in $B_v^s$ indeed correspond to bi-paths between $s$ and $v$.

  At the beginning of the algorithm, all sets are empty except for $B_s^s$ that contains $\{((\bot,-\infty,\infty))\}$, since $s$ has a bi-path of length~$0$ with itself that can start and end at any time, this is a valid triplet.

  Adding triplets to $B_v^s$ is done through the $+$ operator that encapsulates the extension and elimination rules. Since the extension rule guarantees that the extended bi-path is valid if the previous one is (Lemma~\ref{lem:extension}), this means that $B_v^s$ contains only valid bi-paths.

  Assume now that there exists a bi-path $b_{s,v}$ between~$s$ and~$v$ that was not computed by such algorithm, such a bi-path can be seen as a sequence of triplets 
  $$((\bot,-\infty,\infty),(u_1,a_1,d_1),\dots,(u_k,a_k,d_k))$$
  such that $a_i$ is the arrival time at $u_i$ of journey $j_1$ starting at~$s$ and $d_i$ is the departure times from $u_i$ of journey~$j_2$ going to~$s$.
  
  By induction on the length of such bi-path:
  If $b_{s,v}$ has length $k=0$, then it is computed by the algorithm since $\{((\bot,-\infty,\infty))\}$ is part of $B_s^s$.
  Suppose that $b_{s,v}$ has now length $k>0$ and up to $(u_{k-1},a_{k-1},d_{k-1})$, it is part of a bi-path computed by the algorithm, that is, there is a triplet $(x,a_{k-1}',d_{k-1}')$ (where $x$ is either $u_{k-2}$ or $\bot$) in $B_{u_{k-1}}^s$ such that $a_{k-1} \geq a_{k-1}'$ and $d_{k-1} \leq d_{k-1}'$.
  Such triplet, by the time it was added to $B_{u_{k}}^s$, would imply a call to \texttt{improveNeighbor($u_{k}$)}, and by the extension rule, a triplet $(u_{k},a,d)$ in $B_v^s$ such that
  $a = \min \{t\in \lambda(u_{k}v) \mid t \ge a_{k-1}'\}$ and
  $d = \max \{t\in \lambda(u_{k}v) \mid t \le d_{k-1}'\}$.

  Since $a$ is the smallest $t \ge a_{k-1}'$ in $\lambda(u_{k-1}v)$ and $a_{k-1} \geq a_{k-1}'$, this means that the label $a$ is either also the smallest label which is larger or equal to~$a_{k-1}$ or is a label even smaller than such label.
  Since $a_v \in \lambda(u_{k-1}v)$, that is  at  least~$a_{k-1}$ (otherwise $j_1$ would not be a journey), we have that $a \leq a_v$. By the similar argument, we also have $d \geq d_v$. Thus $(u_{k},a_v,d_v)$ is included into a triplet of $B_v^s$, and the bi-path $b_{s,v}$ is computed by the algorithm, a contradiction.
  \lncsqed
\end{proof}

\begin{theorem}
  Algorithm~\ref{alg:bidirectional} runs in polynomial time in the size of the input, namely in time $O(m \Delta^2 \tau^4)$.
\end{theorem}

\begin{proof}
  \label{cor:polynomial}
  By the elimination rule, for each vertex $u$, $B_u^s$ contains less than $|N(u)| \cdot \tau$ triplets since given a neighbor and an arrival time, at most one departure time should be kept in $B_u^s$.
  Moreover, by the same rule, $B_u^s$ can be improved at most $|N(u)| \cdot \tau^2$ times.

  This means that the algorithm will be called at most $m \cdot \tau^2$ times.
  Each call of the function costs $|B(u)| \cdot |N(u)| \cdot O(\tau)$~time where $O(\tau)$ upper bounds the cost of finding the minimum or maximum label in the given range, plus the cost of applying the elimination rule and comparing it, meaning that a call costs $O(|N(u)|^2 \cdot \tau^2)$~time overall. Thus, the total running time is $O(m \Delta^2 \tau^4)$~time. \lncsqed
\end{proof}

To then solve~\BC, it suffices to consider each vertex as source~$s$ and check whether there is a bi-path between~$s$ and any other vertex of the graph.
We conclude the following.

\begin{theorem}\label{bs poly}
One can decide in polynomial time whether a given temporal graph contains a bi-spanner, that is, \BC can be solved in polynomial time.
\end{theorem}

\section{Bi-spanners of small size}
\label{sec:kbs}

In this section, we discuss several aspects of the size of bidirectional spanners. We start with a negative structural result regarding the minimum size of such spanners. This case being extreme, we investigate the problem of finding a bi-spanner of a certain size $k$. We already know, from the earlier sections, that this problem is hard, because spanning trees are a particular case. In fact, we strengthen this result by showing that \BS is hard even in simple temporal graphs (where spanning trees were tractable). This motivates future work for identifying special cases by other means (e.g., specific parameters).

Recall, from the introduction, that sparse temporal spanners (let apart bidirectionality) are not guaranteed to exist in temporal graphs, as there exist temporal graphs with $\Theta(n^2)$ edges, all of which are critical for temporal connectivity~\cite{AF16}. However, if the footprint is a \emph{complete} graph and the setting is either non-strict or proper, then spanners of size $O(n \log n)$ exist no matter what the labeling is~\cite{CPS19}.
The following shows that no analog result exists for bidirectional spanners.

\begin{lemma}
  \label{lem:clique}
  Let $\G=((V,E),\lambda)$ be a happy clique, then $\G$ is bidirectionally temporally connected, but for every $e \in E, \G \setminus \{e\}$ is not temporally connected.
\end{lemma}
\begin{proof}
  The fact that $\G$ is bidirectionally temporally connected is straightforward, as every edge in $\G$ forms a bi-path between its two endpoints and the footprint is a complete graph. Now, if an edge $\{u,v\} \in E$ is removed, then the only way $u$ and $v$ can reach each other is through a temporal path of length at least $2$. Since the labeling is happy, the labels along such a journey must increase, thus the path is not reversible. \lncsqed
\end{proof}

Strictly speaking, Lemma~\ref{lem:clique} makes the algorithm of Section~\ref{sec:bs} worst-case optimal; however, for reasons that pertain to bidirectionality rather than to the algorithm itself, which is not satisfactory. A natural question is whether small bi-spanners can be found when they exist. In the following, we show that the answer is negative even in simple temporal graphs. We do so in the non-strict setting, since bi-spanner in simple temporal graphs in the strict setting are always cliques (see Lemma~\ref{lem:clique}), which is easy to test.

\begin{theorem}
  \label{th:bispanner-hard}
  \BS is NP-complete in the simple and non-strict setting.
\end{theorem}

\begin{proof}
  Given a candidate spanner $\mathcal{G}' \subseteq \mathcal{G}$, it is easy to check whether $\G'$ uses at most $k$ edges and whether it is bidirectionally connected (using the algorithm of Section~\ref{sec:bs}), thus the problem is in NP. To show that this is NP-hard, we reduce from \textsc{Set Cover}, defined as follows: given a set~$\mathcal{U}$ of $n$ elements (the universe) and a collection ${\cal S}=\{S_1, \dots,S_m\}$ of $m$ subsets of $\mathcal{U}$ (whose union is $\mathcal{U}$), is there a union of at most $k$ sets of~${\cal S}$ that is equal to $\mathcal{U}$?

  Let $\mathcal{I}$ be an instance of \textsc{Set Cover}, we construct a temporal graph $\G=((V,E),\lambda)$ with $n+m+3$ vertices such that $\mathcal{I}$ admits a solution of size $k$ if and only if a bi-spanner of size $3n + 2m + 3 + k$ exists in $\G$.
  The vertices $V$ are:
  \begin{itemize}
    \item one vertex $u_i$ for each element of $\mathcal{U}$;
    \item one vertex $s_i$ for each subset $S_i$;
    \item three auxiliary vertices $v, v_1,v_2$ that are used for connectivity.
  \end{itemize}

  Since $\mathcal{G}$ is simple and in the non-strict setting, each bidirectional journey must use edges from the same snapshot, thus the order of the time labels is irrelevant, what matters is only which sets of vertices are connected in a same snapshot.
  The edges are constructed as follows (see Figure~\ref{fig:spanner-btc} for an illustration):
  \begin{itemize}
  \item $s_i$ shares an edge at time $i$ with $v$ and with each $u_j \in S_i$;
    \item $v_1$ shares an edge at time $m+1$ with all vertices~$s_i$ with $S_i\in {\cal S}$ and with all $u_j\in \mathcal{U}$;
    \item $v_2$ shares an edge at time $m+2$ with $v$ and with all $s_i\in {\cal S}$;
    \item $v_1$ and $v_2$ share an edge with all the other nodes (and with each other) at unique arbitrary times.
  \end{itemize}

  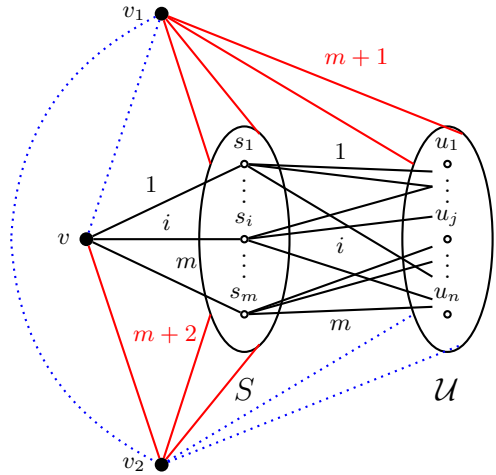
\begin{figure}[!ht]
    \centering
    \begin{tikzpicture}[every loop/.style={}, thick, node distance=1.2cm,auto]
      \tikzset{vertex/.style={draw, circle, inner sep=0.55mm,fill=black}}
      \tikzset{invertex/.style={draw, circle, inner sep=0.3mm}}
      \tikzstyle{sameedge} = [draw,red]
      \tikzstyle{properedge} = [draw,dotted,blue]
      \tikzstyle{staredge} = [draw]

      \node[vertex] (v) at (0,0) [label=left:$v$] {};
      \node[vertex] (u1) at (1,3) [label=left:$v_1$] {};
      \node[vertex] (u2) at (1,-3) [label=left:$v_2$] {};

      \node[invertex] (s1) at (2.1,1) [label=$s_1$] {};
      \node at (2.1,0.75) {\vdots};
      \node[invertex] (si) at (2.1,0) [label=$s_i$] {};
      \node at (2.1,-0.25) {\vdots};
      \node[invertex] (sm) at (2.1,-1) [label=$s_m$] {};

      \node[invertex] (x1) at (4.8,1) [label=$u_1$] {};
      \node at (4.8,0.75) {\vdots};
      \node[invertex] (xj) at (4.8,0) [label=$u_j$] {};
      \node at (4.8,-0.25) {\vdots};
      \node[invertex] (xn) at (4.8,-1) [label=$u_n$] {};

      \draw (2.1,0) ellipse (0.6cm and 1.5cm);
      \node (S) at (2.1,-2) {\large ${\cal S}$};
      \draw (4.8,0) ellipse (0.6cm and 1.5cm);
      \node (X) at (4.8,-2) {\large $\mathcal{U}$};

      \begin{scope}
        \draw (v) edge[staredge] node {$1$} (s1);
        \draw (s1) edge[staredge] node {$1$} (4.6,0.9);
        \draw (s1) edge[staredge] (4.6,0.7);
        \draw (s1) edge[staredge] (4.6,-0.5);
        \draw (v) edge[staredge] node {$i$} (si);
        \draw (si) edge[staredge] (4.6,0.7);
        \draw (si) edge[staredge] node[below] {$i$} (4.6,0.3);
        \draw (si) edge[staredge] (4.6,-0.8);
        \draw (v) edge[staredge] node {$m$} (sm);
        \draw (sm) edge[staredge] (4.6,-0.1);
        \draw (sm) edge[staredge] (4.6,-0.3);
        \draw (sm) edge[staredge] node[below] {$m$} (4.6,-0.9);

        \draw (v) edge[sameedge] node {$m+2$} (u2);
        \draw (u2) edge[sameedge] (1.65,-1);
        \draw (u2) edge[sameedge] (2.3,-1.4);
        \draw (u2) edge[properedge] (4.35,-1);
        \draw (u2) edge[properedge] (5,-1.4);
        \draw (v) edge[properedge] (u1);
        \draw (u1) edge[sameedge] (1.65,1);
        \draw (u1) edge[sameedge] (2.3,1.4);
        \draw (u1) edge[sameedge] (4.35,1);
        \draw (u1) edge[sameedge] node {$m+1$} (5,1.4);
        \draw (u1) edge[properedge, bend right, bend angle = 60] (-1,0);
        \draw (-1,0) edge[properedge, bend right, bend angle = 60] (u2);
      \end{scope}
    \end{tikzpicture}
    \caption{The reduction, dotted edges have a distinct label from their neighborhood}
    \label{fig:spanner-btc}
  \end{figure}

  \begin{nclaim}
    $\G$ is bidirectionally temporally connected.
  \end{nclaim}
  \begin{proof}
    To start, observe that $v_1$ and $v_2$ both share an edge with all the other vertices and with each other, they are thus bi-connected to all vertices. Moreover, each vertex of ${\cal S}$ is bi-connected to $v$ and to other vertices in ${\cal S}$ through $v_2$, and to all vertices of $\mathcal{U}$ through $v_1$. Vertices of~$\mathcal{U}$ are also bi-connected to each other through $v_1$. Finally, for each $u_j\in \mathcal{U}$, at least one set $S_i$ contains it, so $u_j$ is bi-connected to $v$ at the corresponding time $i$.\lncsqed
  \end{proof}

  Now, observe that $v_1$ and $v_2$ are only bi-connected to other vertices through bi-paths of length $1$, thus none of their incident edges can be removed. This implies that any bi-spanner must use at least these $2|V|-3 = 2n + 2m + 3$ edges (twice $|V|-1$, minus their common edge). These edges will always bi-connect all the vertices in $\mathcal{U}$ with each other, all the vertices of ${\cal S}$ with each other, and all the vertices in $\mathcal{U}$ with all the vertices in ${\cal S}$.

  ($\Rightarrow$) If there is a solution of size $k$ to the \textsc{Set Cover} instance~$I$, then a bi-spanner of size $3n+2m+k+3$ can be built as follows: for each $s_i$ selected in the solution, keep the edge $vs_i$ and remove all edges~$vs_j$ for which the set~$S_j$ is not part of the solution. For each $u_j$, keep an edge to one of the $s_i$ for which set~$S_i$ contains~$u_j$ and that is selected in the solution. Remove all other edges between the~$s_i$ and~$u_j$ vertices. Then all $u_j$ are still bi-connected to $v$. This makes $k+n$ edges, plus the above required edges, which is $3n+2m+k+3$ in total.

  ($\Leftarrow$) If a bi-spanner of size $3n+2m+k+3$ exists, then it contains at least the above required edges, so in total only $k+n$ edges overall are kept between $v$ and vertices of ${\cal S}$ and between vertices of ${\cal S}$ and vertices of $\mathcal{U}$. The fact that each vertex in $\mathcal{U}$ remains bi-connected to $v$ implies that at least one edge is kept for each of them to a vertex~$s_i$ for some set~$S_i\in {\cal S}$, which makes at least $n$ edges, and that the corresponding vertex~$s_i$ still shares an edge with $v$. There are at most $k$ such edges, thus there exists a subset of ${\cal S}$ of size at most $k$ that covers all the nodes in $\mathcal{U}$.
  \lncsqed
\end{proof}

\section{Parameterized algorithms for tree-like graphs}\label{sec:fes}

In this section, we show that \TST is FPT for parameter~$\fes$ and \BS is FPT for parameter combination~$\fes + \ell$, where~$\fes$ denotes the feedback edge set number of the underlying graph and~$\ell$ denotes the maximum number of labels of any edge.
Here, the \emph{feedback edge set number} $\fes$ of a static graph~$G$ is defined as the size of the smallest edge set~$F$, such that~$G-F$ is acyclic.
Such an edge set is also called a~\emph{feedback edge set}.
To present both of our algorithms, we use some preprocessing steps and definitions from the literature~\cite{EMM25}.

\begin{definition}
The 2-core of a graph~$G$ is the unique largest vertex set~$V^*$, such that each vertex of~$V^*$ has at least two neighbors in~$V^*$.
\end{definition}

Note that in both~\TST and \BS, each solution has to keep all edges that have at most one endpoint in~$V^*$, as each such edge is an edge-cut of size~1 of~$G$.
Thus, we only have to decide which edges we keep and remove that have both their endpoints in~$V^*$.
We show that we can decompose the edges of~$V^*$ in few parts as follows.
Let~$X\subseteq V^*$ such that~$X \supseteq \bigcup \{N[v]\mid v\in V^*, |N(v)| \geq 3\}$, that is, $X$ contains (at least) all vertices of degree at least~3 in~$G[V^*]$ and their neighbors.

\begin{definition}
Let~$X\subseteq V^*$ as specified above.
A path~$C$ in~$G[V^*]$ is an~\emph{$X$-connector} if (i)~the endpoints of~$C$ are in~$X$ and non-adjacent with each other and (ii)~all vertices of~$C$ have degree exactly~$2$ in~$G[V^*]$.
Moreover, the \emph{extension of~$C$} is the tree~$T$ obtained from~$C$ by adding all vertices of~$V\setminus V^*$ for which the closest neighbor in~$V^*$ is a vertex of~$V(C)\setminus X$.
\end{definition}

Erlebach et al.~\cite{EMM25} showed that one can compute a small set of vertices~$X$ with some useful properties.

\begin{lemma}[\cite{EMM25}]
In polynomial time, one can compute a set~$X$ with~$\bigcup \{N[v]\mid v\in V^*, |N(v)| \geq 3\}\subseteq X\subseteq V^*$, such that (i)~$X$ has size~$\Oh(\fes)$, (ii)~there are~$\Oh(\fes)$ edges between vertices of~$X$, (iii)~there are~$\Oh(\fes)$ many~$X$-connectors, and (iv)~each edge of~$G[V^*]$ is either between vertices of~$X$ or part of exactly one~$X$-connector.
\end{lemma} 

In the following, let~$X$ be a vertex set with these properties that we initially compute in polynomial time.
Note that for each~$X$-connector~$C$, we can remove at most one edge of~$C$, as we would increase the number of connected components by~1 for whichever two edges we remove of~$C$.
Thus, in both~\TST and \BS, we can only decide for each~$X$-connector to remove either no edge of~$C$ or exactly one edge of~$C$.
Let~$\mc$ denote the set of all~$X$-connectors.
As~$|\mc| \in \Oh(\fes)$ and there are~$\Oh(\fes)$ many edges~$E(X)$ between the vertices of~$X$, in both of our algorithms, we initially branch in all possibilities~$(E^*,\mc^*)$ with~$E^*\subseteq E(X)$ and~$\mc^*\subseteq \mc$ and check whether there is a solution that agrees with the choice of~$(E^*,\mc^*)$, that is, a solution that (i)~contains exactly the edges~$E^*$ of~$E(X)$, (ii)~contains all edges of the~$X$-connectors of~$\mc^*$ and their endpoints, and (iii)~contains all but exactly one edge of each~$X$-connectors in~$\mc\setminus \mc^*$.
As these are~$2^{\Oh(\fes)}$ possibilities, we can afford this branching, as both of our algorithms will be parameterized by the feedback edge set number.
For a given choice~$(E^*,\mc^*)$, we call the subgraph containing exactly the edges of~$E^*$ and the edges of all~$X$-connectors of~$\mc^*$ a~\emph{backbone}.
Note that if the backbone is not connected, then this choice cannot lead to a solution, as the only way to connect different connected components would go over~$X$-connectors in~$\mc\setminus \mc^*$, that is, over $X$-connectors for which we chose to remove exactly one edge.
Another exclusion criteria is if the feedback edge set number of the backbone is larger than~$k-(n-1)$, in this case, the remaining budget would not be enough to obtain a connected graph from the backbone that contains all vertices of~$V$.   
In particular, this implies that for~\TST, the backbone is a tree, or the choice is not valid.
For the remainder of the section, we describe our algorithms for a given choice~$(E^*,\mc^*)$.
Essentially in both algorithms, we need to determine which concrete edge of each of the connectors of~$\mc\setminus \mc^*$ we can safely remove while preserving bi-paths over the backbone to all other vertices.
This choice of which edge to remove per connector needs to be consistent over all connectors.
We now distinguish between the case of~\TST and \BS.

\subsection{The algorithm for \TST}
We start with our algorithm for~\TST. Here, we make use of the fact that each solution will have a pivot contact or pivot vertex, which we can determine by enumerating all possible candidates for them.

\begin{theorem}
\TST can be solved in $2^{\Oh(\fes)}\cdot n^{\Oh(1)}$~time.
\end{theorem}
\begin{proof}
Recall that we can enumerate all candidate backbones in $2^{\Oh(\fes)}\cdot n^{\Oh(1)}$~time.
To prove the algorithm, we show that, a given backbone~$G'$ for a choice~$(E^*,\mc^*)$, we determine in polynomial time, whether for each~$X$-connector of~$\mc\setminus \mc^*$, we can keep all but exactly one edge of that~$X$-connector to preserve temporal connectivity.
To show this result, we exploit the fact that in every \TC spanning tree there is a pivot contact or a pivot vertex~\Cref{pivot in trees}.
This allows us to show that the choice of which edge of the connector to remove has no influence on the choices of the other connectors.

Assume that there is a solution~$\G'$ to our problem that extends the given backbone~$G'$.
Then, this solution contains a pivot contact or a pivot vertex due to~\Cref{pivot in trees}.
As the number of contacts and temporal vertices is linear in the input size, we can iterate efficiently over all candidate contacts~$(uv,t)$ and candidate temporal vertices~$(v,t)$ and check whether there is a solution that agrees with the choice of the backbone and uses~$(uv,t)$ as the pivot contact or~$(v,t)$ as a pivot vertex.
In the remainder, we only discuss the case for pivot contacts. The case for pivot vertices can be handled analogously.
Assume for simplicity that the edge~$uv$ is not part of the extension of any $X$-connector in~$\mc\setminus \mc^*$.
(We argue the case where this is the case at the end.)
The problem thus boils down to deciding whether there is one edge per $X$-connector in~$\mc\setminus \mc^*$, such that removing that edge from the connector still allows all vertices in both resulting sides of the extension of the connector to reach~$uv$ prior to time~$t$ and be reached by both endpoints of~$uv$ starting aftertime~$t$.
As there are only $\Oh(m)$~edges in each connector~$C$, we can decide this in polynomial time, by iterating over all possible edges to remove and afterwards checking whether the earliest arrival and latest departure times towards the desired pivot contact are preserved.
If this is possible for at least one edge of the connector, we continue with the next connector and answer YES if we considering the last connector.
Otherwise, we continue with the next candidate for the pivot contact, as there is no possible way to cut at least one connector~$C$ to preserve the property of~$uv$ being a pivot contact at time~$t$. 
In the case where the pivot contact is part of the extension of some~$X$-connector in~$\mc\setminus \mc^*$, we handle this $X$-connector~$C$ initially in a separate way.
As there has to be one edge~$e\neq uv$ removes from~$C$, we can also initially iterate over all choices of which concrete edge of~$C$ to remove.
For each such choice, we remove the edge~$e$ from our instance and treat the two resulting parts of the connector as parts of the backbone.
For all other connectors, we proceed as before with this updated backbone.

The algorithm is correct, since if there is a solution to the problem, this solution has to agree with some choice of the backbones we iterated over (as we consider all possibilities).
Moreover, the solution contains a pivot structure due to~\Cref{pivot in trees} and we iterate over all possible such pivot structures.

Finally, we obtain the stated running time, since 
\begin{itemize}
\item we iterate initially over~$2^{\Oh(\fes)}$ choices for the backbone,
\item we iterate over linearly many choices for the pivot structure,
\item we potentially iterate for one dedicated connector over all its edges, and
\item for each other connector in~$\mc \setminus \mc^*$, we independently check whether there is some edge to remove to fulfill the properties of the pivot structure in the solution.
\end{itemize}
As all other checks run in polynomial time, we obtain a total running time of $2^{\Oh(\fes)}\cdot n^{\Oh(1)}$~time.
\lncsqed
\end{proof}

\subsection{The algorithm for \BS}

Finally, we present our parameterized algorithm for \BS.
Again, we try to extend a backbone to a solution to the problem.
Unfortunately, this time (that is, for \BS), the choices of which edges to remove in each connector is not independent between the individual connectors as it was for \TST.
This is due to the fact that there does not necessarily exist a pivot structure in the case, where the temporal spanner is not a tree.
However, if we include the number of labels per edge in our parametrization, we can still obtain an FPT-algorithm.  
Let $\ell := \max_{e\in E} |\lambda(e)|$ denote the maximum number of labels per edge.

\begin{theorem}
\BS can be solved in $\max(2,\ell)^{\Oh(\fes)}\cdot n^{\Oh(1)}$~time, where~$\ell$ denotes the maximum number of labels per edge.
\end{theorem}
\begin{proof}
Recall that we can enumerate all candidate backbones in $2^{\Oh(\fes)}\cdot n^{\Oh(1)}$~time.
To prove the algorithm, we show that, a given backbone~$G'$ for a choice~$(E^*,\mc^*)$, we determine in $\max(2,\ell)^{\Oh(\fes)}\cdot n^{\Oh(1)}$~time whether for each~$X$-connector of~$\mc\setminus \mc^*$, we can keep all but exactly one edge of that~$X$-connector to preserve the property of being a bi-spanner.
To show this result, we essentially show that we only need to consider $\Oh(\ell^4)$ possible edges to cut for each of the connectors. 
The total algorithm then branches over the $\Oh(\ell^4)$ choices for all connectors simultaneously and checks whether the resulting temporal graph is still a bi-spanner by using the algorithm behind~\Cref{bs poly}.

Consider the extension~$Q$ of any~$X$-connector~$C$ of~$\mc\setminus \mc^*$.
Let~$v_1$ and~$v_2$ be the endpoints of~$C$ and let~$e_1$ and~$e_2$ be the first and last edge of~$C$ respectively.
For each edge~$e$ of~$C$ distinct from~$e_1$ and~$e_2$ and each~$i\in \{1,2\}$, let~$Q_i$ denote the connected component of~$G[Q] - e$ that contains~$v_i$.
We will define the `relevant' labels of edge~$e_i$ after the removal of~$e$ as follows:
We set~$\alpha_{i}^e$ to be the smallest label of edge~$e_i$ such that each vertex of~$Q_i$ can traverse edge~$e_i$ at time~$\alpha_{i}^e$ in~$\G-e$.
If no such label $\alpha_{i}^e$ exists, then cutting edge~$e$ would destroy temporal connectivity. 
Similarly, we set~$\omega_{i}^e$ to be the largest label of edge~$e_i$ such that each vertex of~$Q_i$ can be reached after traversing edge~$e_i$ at time~$\omega_{i}^e$ in~$\G-e$.
Again, if no such label $\omega_{i}^e$ exists, then cutting edge~$e$ would destroy temporal connectivity.

Now consider the set of label tuples~$\mathcal{L(C)} := \lambda(e_1) \times \lambda(e_1) \times \lambda(e_2) \times \lambda(e_2)$. 
As each edge has at most~$\ell$ labels, $L(C)$ has size at most~$\ell^4$.
For each~$\pi:=(\alpha_1,\omega_1,\alpha_2,\omega_2)\in \mathcal{L}$, we essentially only need to consider a single candidate edge of~$C$ to remove (if one exists).
That is, we define~$e_\pi$ to be any edge of~$C$, such that~$\G-e_\pi$ still has a bi-spanner and such that for each~$i\in \{1,2\}$, $\alpha_i^{e_\pi} = \alpha_i$ and~$\omega_i^{e_\pi} = \omega_i$.
If such an edge~$e_\pi$ exists, we add it to the set~$R(C)$, which also contains the edges~$e_1$ and~$e_2$.

We now argue that it suffices to branch only on the edges of~$R(C)$ to find a fitting edge to remove from~$C$, if there is any edge that can be removed.
Assume that there is a solution~$\G^*$ that extends the backbone and that cuts for some connector~$C\in \mc\setminus \mc^*$ an edge~$e$ of~$C$ that is not in~$R(C)$.
Moreover, each vertex of~$Q$ needs to be able to traverse both edges~$e_1$ and~$e_2$ in~$\G^*$ via temporal paths.
Hence all values of~$\pi = (\alpha_1^e,\omega_i^e,\alpha_2^e,\omega_2^e)$ are defined.
As~$\pi\in \mathcal{L}(C)$ and $\G-e$ has a bi-spanner by necessity, we conclude that the edge~$e_\pi$ is defined.
Hence, adding the edge~$e$ back to the solution and cutting~$e_\pi$ instead also yields a solution, as all vertices of the extension~$Q$ of~$C$ can still traverse the edges~$e_1$ and~$e_2$ at the exactly same times as before.
Hence, we obtain a solution that has less edges outside of~$R(C)$ for all connectors.
Thus, applying this argument inductively, we get a solution, where each edge that is cut from any connector~$C$ is part of the respective set~$R(C)$.

Based on this fact, we now describe the algorithm formally.
For each choice of the backbone, the algorithm computes the set~$R(C)$ for each connector~$C\in \mc\setminus \mc^*$.
Afterwards, one iterates over all possible combinations of edges over all~$R(C)$ sets.
That is, over~$R(C_1)\times \dots \times R(C_{|\mc\setminus \mc^*|})$.
For each such combination, we remove all the respective edges from the temporal graph and check whether the resulting temporal graph obtained from the backbone and all remaining edges of the extensions of connector paths still contains a bi-spanner.
If this is the case, we answer YES. 
Otherwise, if this fails for all possible combinations of edges, we continue with another choice for the backbone.
If this procedure does not answer YES for at least one backbone, we eventually answer NO.

By the above argumentation, the algorithm is correct.
It remains to show the running time.
We iterate over $2^{\Oh(\fes)}$ backbones, for each such backbone, we compute all the sets~$R(C)$ in polynomial time, and afterwards iterate over all possible combinations of edges.
As~$|R(C)|\in \Oh(\ell^4)$ and there are $\Oh(\fes)$~connectors, we iterate over $\max(2,\ell)^{\Oh(\fes)}$ such edge combinations.
For each such combination, we then check whether the resulting temporal graph still contains a bi-spanner in polynomial time due to the algorithm behind~\Cref{bs poly}. 
 This completes the proof that \BS can be solved in $\max(2,\ell)^{\Oh(\fes)}\cdot n^{\Oh(1)}$~time.\lncsqed
\end{proof}

\section{Conclusion}
\label{sec:conclusion}

In this paper, we answered a fundamental question in temporal graphs, namely, whether deciding the existence of a \TC spanning tree is tractable. It was already known that finding temporal spanners of minimum size is a hard problem. Unfortunately, we showed that even in the extreme case that the spanner is a tree, the problem is NP-complete. Our reduction rely on a proper temporal graph, and as such, applies to both strict and non-strict journeys. Motivated by this negative result, we explored different ways to relax spanning trees in temporal graphs. In particular, we showed that the property of bidirectionality without a prescribed size can be tested efficiently, but it is harder in than deciding \TC spanning trees otherwise. Still, we believe that testing bidirectionality could prove useful in certain applications, and this adds to a small collection of non-trivial properties that remain tractable in temporal graphs. As already discussed, such applications might be related to routing or security. They also include transportation networks, where bus or train lines benefit from being operated along the same segments in both directions.
Beyond these basic results, we presented two FPT algorithms, one for each problem, where the parameter for \TC spanning tree is \fes and the parameters for bidirectional spanner are \fes + the maximum number of labels per edge.
Finally, of independent interest, we proved that \TC trees always contain pivot contacts or pivot vertices.
Some questions remain open, in particular, what additional restrictions could force the existence of small bidirectional spanners, and what parameters could make their computation tractable.
\begin{question}
  If the bi-path distance between pairs of vertices is large on average (over all the pairs), does this guarantee the existence of a small bidirectional spanner? Could such distances be large in a dense footprint?
\end{question}

Beyond structural questions, a natural approach for tractability is to explore further parameters with respect to which the problem is fixed-parameter tractable.

\begin{question}
  For which other parameters do \BS and \TST admit FPT algorithms? In particular, is \BS FPT for $\fes$ alone and does \TST admit an FPT algorithm for a parameter smaller than~$\fes$, for example the treewidth of the footprint.
\end{question}

\bibliographystyle{plainurl}
\bibliography{main.bib}

\end{document}